\documentclass{article} % For LaTeX2e
\PassOptionsToPackage{usenames,dvipsnames,dvinames}{xcolor}
\usepackage[usenames,dvipsnames,dvinames]{xcolor}
\usepackage{multirow}
\usepackage{microtype}
\usepackage{enumitem}
\usepackage{amsmath,amsthm}
\usepackage[mathscr]{eucal}
\usepackage{amssymb}
\usepackage{footnote}
\usepackage{graphicx} % more modern
\usepackage[vlined,ruled]{algorithm2e}
\usepackage{algorithmic}
\usepackage{wrapfig}
\usepackage{float}
\usepackage{caption}
\usepackage{titlesec}

\titlespacing{\paragraph}{%
  0pt}{%              left margin
  0.3\baselineskip}{% space before (vertical)
  1em}%   

\DeclareMathOperator*{\argmax}{argmax}
\DeclareMathOperator*{\argmin}{argmin}

\newtheorem{theorem}{Theorem}[section]

\newtheorem{lemma}[theorem]{Lemma}

\newtheorem{corollary}[theorem]{Corollary}
\newcommand{\curvf}[1]{\ensuremath{\kappa_{#1}}}
\newcommand{\eax}[1]{\ensuremath{\sqrt{w^{#1}}}}
\newcommand{\eaxa}[2]{\ensuremath{\sqrt{w^{#1}(#2)}}}
\newcommand{\perm}{\ensuremath{\pi}}
\providecommand{\doarxiv}{true}
\usepackage{ifthen}
\newboolean{isarxiv}
\setboolean{isarxiv}{\doarxiv} 
\ifthenelse{\boolean{isarxiv}}{%
\newcommand{\arxiv}[1]{#1}
\newcommand{\notarxiv}[1]{}
}{
\newcommand{\arxiv}[1]{}
\newcommand{\notarxiv}[1]{#1}
}
\newcommand{\arxivalt}[2]{\ifthenelse{\boolean{isarxiv}}{#1}{#2}}
\newcommand{\arxivaltr}[2]{\ifthenelse{\boolean{isarxiv}}{#2}{#1}}
\newcommand{\narxiv}[1]{\notarxiv{#1}}

\newcommand{\truncated}{curve-normalized}

\usepackage{ifthen}
\newboolean{isdraft}
\setboolean{isdraft}{false} 
\ifthenelse{\boolean{isdraft}}{%
\newcommand{\myaddcomment}[3]{{\color{#1}{\ensuremath{\langle\!\!\langle}{\bf {#2} :} {#3}\ensuremath{\rangle\!\!\rangle}}}}
\newcommand{\rishabh}[1]{\myaddcomment{orange}{Rishabh}{#1}}
\newcommand{\JTR}[1]{\myaddcomment{orange}{Jeff\ensuremath{\rightarrow}Rishabh}{#1}}
\newcommand{\jeff}[1]{\myaddcomment{blue}{Jeff}{#1}}
\newcommand{\RTJ}[1]{\myaddcomment{blue}{Rishabh\ensuremath{\rightarrow}Jeff}{#1}}
}{
\newcommand{\rishabh}[1]{}
\newcommand{\JTR}[1]{}
\newcommand{\jeff}[1]{}
\newcommand{\RTJ}[1]{}
\newcommand{\toboth}[1]{}
}
\title{Submodular Optimization with Submodular Cover and Submodular Knapsack Constraints}

\author{ {\bf Rishabh Iyer} \\
Dept. of Electrical Engineering\\  
University of Washington\\ 
Seattle, WA-98175, USA
\and 
{\bf Jeff Bilmes} \\ 
Dept. of Electrical Engineering \\  
University of Washington\\ 
Seattle, WA-98175, USA
} 
\begin{document}

\maketitle

\begin{abstract}
  We investigate two new optimization problems --- minimizing a
  submodular function subject to a submodular lower bound constraint
  (submodular cover) and maximizing a submodular function subject to a
  submodular upper bound constraint (submodular knapsack). We are
  motivated by a number of real-world applications in machine learning
  including sensor placement and data subset selection, which require
  maximizing a certain submodular function (like coverage or
  diversity) while simultaneously minimizing another (like cooperative
  cost). These problems are often posed as minimizing the difference
  between submodular functions~\cite{rkiyeruai2012,narasimhanbilmes}
  which is in the worst case inapproximable. We show, however, that by
  phrasing these problems as constrained optimization, which is more
  natural for many applications, we achieve a number of bounded
  approximation guarantees. We also show that both these problems are
  closely related and an approximation algorithm solving one can be
  used to obtain an approximation guarantee for the other. We provide
  hardness results for both problems thus showing that our
  approximation factors are tight up to $\log$-factors. Finally, we
  empirically demonstrate the performance and good scalability
  properties of our algorithms.\looseness-1
\end{abstract}

%%%%%%%%%%%%%%%%%%%%%%%%%%%%%%%%%%%%%%%%%%%%%%%%%%%%%%%%%%%%%%%%%%%%%%%
\section{Introduction}
\label{sec:introduction}

A set function $f: 2^V \to \mathbb{R}$ is said to be \emph{submodular}
\cite{fujishige2005submodular} if for all subsets $S, T \subseteq V$,
it holds that $f(S) + f(T) \geq f(S \cup T) + f(S \cap T)$.  Defining
$f(j | S) \triangleq f(S \cup j) - f(S)$ as the gain of $j\in V$ in
the context of $S \subseteq V$, then $f$ is submodular if and only
if %\emph{Iff}
$f(j | S) \geq f(j | T)$ for all $S \subseteq T$ and $j \notin T$. The
function $f$ is monotone iff $f(j | S) \geq 0, \forall j \notin S, S
\subseteq V$.  For convenience, we assume the ground set is $V = \{1,
2, \cdots, n\}$. While general set function optimization is often
intractable, many forms of submodular function optimization can be
solved near optimally or even optimally in certain cases\arxivalt{,
  and hence submodularity is also often called the discrete analog of
  convexity~\cite{lovasz1983}.}{.} Submodularity, moreover, is 
inherent in a large class of real-world applications, particularly in
machine learning, therefore making them extremely useful in practice.\looseness-1

In this paper, we study a new class of discrete
optimization problems that have the following form:
%\arxivalt{
%%\begin{center}
%\fbox{\begin{minipage}{0.45\textwidth}
%\begin{alignat}{2}
%\label{eqn:scsc_def}
%\text{Problem 1 (SCSC): }\\
%     \underset{X}{\text{minimize }}        & f(X) \notag \\
%    \text{subject to }     & g(X) \geq c \notag
%  \end{alignat}
%\end{minipage}}%
%\hspace{0.051\textwidth}
%\fbox{\begin{minipage}{0.45\textwidth}
%  \begin{alignat}{2}
%\label{eqn:scsk_def}
%\text{Problem 2 (SCSK): }\\
%    \underset{X}{\text{maximize }}    & g(X) \notag \\
%    \text{subject to }     & f(X) \leq b \notag
%  \end{alignat}
%\end{minipage}}\\[1ex]
%end{center}
%}{
\begin{align*}
\label{eqn:scsc_scsk_def}
\mbox{Problem 1 (SCSC): } \min \{f(X) \, | \,g(X) \geq c\},\quad
\text{and}
\quad
\mbox{Problem 2 (SCSK): } \max \{ g(X) \,| \,f(X) \leq b\},
\end{align*}
%}
where $f$ and $g$ are monotone non-decreasing submodular functions
that also, w.l.o.g., are normalized ($f(\emptyset) = g(\emptyset) =
0$)\footnote{A monotone non-decreasing normalized ($f(\emptyset) = 0$) submodular function is called a polymatroid function.}, and where $b$ and $c$ refer to budget and cover parameters
respectively. The corresponding constraints are called the submodular cover~\cite{wolsey1982analysis} and submodular knapsack~\cite{atamturk2009submodular}
 respectively and hence 
we refer to Problem 1 as {\em Submodular Cost Submodular
  Cover} (henceforth SCSC) and Problem 2 as {\em Submodular Cost Submodular
  Knapsack} (henceforth SCSK). Our motivation stems from an interesting class of problems that
require minimizing a certain submodular function $f$ while
simultaneously maximizing another submodular function $g$. We shall
see that these naturally occur in applications like sensor placement,
data subset selection, and many other machine learning applications. A standard approach used in
literature~\cite{rkiyeruai2012, narasimhanbilmes, kawahara2011prismatic} has been to transform these problems into
minimizing the difference between submodular functions (also called DS
optimization):\looseness-1
\begin{align}
\mbox{Problem 0: } \min_{X \subseteq V} \bigl( f(X) - g(X)\bigr).  
\end{align}
While a number of heuristics are available for solving Problem
0\arxiv{, and while these heuristics often work well in
  practice~\cite{rkiyeruai2012,narasimhanbilmes}}, in the worst-case
it is NP-hard and inapproximable~\cite{rkiyeruai2012}, even when $f$ and $g$ are monotone. Although an exact
branch and bound algorithm has been provided for this
problem~\cite{kawahara2011prismatic}, its complexity can be
exponential in the worst case.  On the other hand, in many
applications, one of the submodular functions naturally serves as part
of a constraint.  For example, we might have a budget on a cooperative
cost, in which case Problems 1 and 2 become applicable.
\arxiv{\subsection{Motivation}}
The utility of Problems 1 and 2 become apparent when we consider how
they occur in real-world applications and how they subsume a number of
important optimization problems.\looseness-1

\textbf{Sensor Placement and Feature Selection: }Often,
the problem of choosing sensor locations \arxiv{$A$ from a given set of
possible locations $V$}can be
modeled~\cite{krause2008near, rkiyeruai2012} by
maximizing the mutual information between the chosen variables $A$ and
the unchosen set $V \backslash A$ (i.e.,$f(A) = I(X_A; X_{V
  \backslash A})$). \arxiv{Note that, while the symmetric mutual information is not monotone, it can be shown to approximately monotone~\cite{krause2008near}. }Alternatively, we may wish to maximize the mutual
information between a set of chosen sensors $X_A$ and a quantity
of interest $C$ (i.e., $f(A) = I(X_A ; C)$) assuming that
the set of features $X_A$ are conditionally independent given $C$
\cite{krause2008near, rkiyeruai2012}. Both these functions are submodular. Since there are costs involved, we want to simultaneously minimize the cost $g(A)$. Often this cost is submodular~\cite{krause2008near, rkiyeruai2012}. For
example, there is typically a discount when purchasing sensors in bulk
(economies of scale). \arxivalt{Moreover, there may be diminished cost for
placing a sensor in a particular location given placement in certain
other locations (e.g., the additional equipment needed to install a
sensor in, say, a precarious environment could be re-used for multiple
sensor installations in similar environments). Hence this becomes a
form of Problem 2 above. An alternate view of this problem is to find
a set of sensors with minimal cooperative cost, under a constraint
that the sensors cover a certain fraction of the possible locations,
naturally expressed as Problem 1.}{This then becomes a form of either Problem 1 or 2.\looseness-1 }

\textbf{Data subset selection:} A data subset
selection problem in speech and NLP involves finding a limited
vocabulary which simultaneously has a large coverage. This is
particularly useful, for example in speech recognition and machine
translation, where the complexity of the algorithm is determined by
the vocabulary size. The motivation for this problem is to find the
subset of training examples which will facilitate evaluation of
prototype systems~\cite{lin11}. \arxivalt{This problem then occurs
in the form of Problem 1, where we want to find a small vocabulary
subset (which is often submodular~\cite{lin11}), subject to a
constraint that the subset acoustically spans the entire data set
(which is also often submodular~\cite{lin2009select, linbudget}). This can also be phrased as Problem 2, where we ask for maximizing the
acoustic coverage and diversity subject to a bounded vocabulary size
constraint.}{Often the objective functions encouraging small vocabulary subsets and large acoustic spans are submodular~\cite{lin11, lin2009select} and hence this problem can naturally be cast as an instance of Problems 1 and 2.}
% Possibly we can remove this application below.

% Thu Jul 11 09:25:34 2013 PDT (2013-07-11T09:25:34-0700)
% TODO: to add as application.
% Just thought of an application for these problems. I.e., suppose
% h(X) = I(X;C) where C is an information source you want to preserve
% while e(X) = I(X;D) where D is private information you want to filter out.
% Then finding a subset that, say, maximizes
%      h(X) - e(X) = diff of two submodular functions
% will chose a set of privacy preserving random variables.

% Also, this can be turned into max f s.t. g < \alpha where g is privacy preserving and f is information preserving, and \alpha is a privacy guarantee. 

% I think we should add this app to the NIPS paper if it gets accepted, in fact I'm sorry we didn't add it earlier as it is really important these days.

\textbf{Privacy Preserving Communication: } Given a set of random
variables $X_1, \cdots, X_n$, denote $\mathfrak I$ as an information
source, and $\mathfrak P$ as private information that should be
filtered out. Then one way of formulating the problem of choosing a
information containing but privacy preserving set of random variables
can be posed as instances of Problems 1 and 2, with $f(A) = H(X_A|
\mathfrak I)$ and $g(A) = H(X_A| \mathfrak P)$, where $H(\cdot|\cdot)$
is the conditional entropy. \arxiv{An alternative strategy would be to
  formulate the problem with $f(A) = H(X_A) + H(X_A| \mathfrak I)$ and
  $g(A) = H(X_A) + H(X_A| \mathfrak P)$.}

% h(X) = I(X;C) where C is an information source you want to preserve
% while e(X) = I(X;D) where D is private information you want to filter out.
% Then finding a subset that, say, maximizes
%      h(X) - e(X) = diff of two submodular functions
% will chose a set of privacy preserving random variables.

\textbf{Machine Translation}: Another application in machine
translation is to choose a subset of training data that is optimized
for given test data set, a problem previously addressed with 
modular functions \cite{moore2010intelligent}. Defining a submodular
function with ground set over the union of training and test sample
inputs $V = V_{\text{tr}} \cup V_{\text{te}}$, we can set $f:
2^{V_\text{tr}} \to \mathbb R_+$ to $f(X) = f(X | V_{\text{te}})$, and
take $g(X) = |X|$, and $b \approx 0$ in Problem 2 to address this
problem.  We call this the {\em Submodular Span} problem.

% \JTR{also add submodular span problem, this I think will be very useful
% in apps, and corresponds to an instance of prob 2 with $g$ being modular.}
%\arxiv{\textbf{Summarization: } Often in summarization tasks (like document summarization), we want to simultaneously maximize the coverage of the subsets, while minimizing the similarity in the subset~\cite{linacl}. While many notions of coverage functions are submodular~\cite{linacl}, the similarity between objects can also sometimes be modeled via a submodular function~\cite{borodin2012max}, when defined via a distance measure between objects. This problem becomes an instance of Problems 1 and 2, with $f$ representing the similarity function and $g$ being the coverage.
%\RTJ{Added a new application above -- not sure we will have space though}}

\arxiv{\textbf{Probabilistic Inference:} Many problems in
computer vision and graphical model inference involve finding an
assignment to a set of random variables. The most-probable explanation
(MPE) problem finds the assignment that maximizes the probability. In
computer vision and high-tree-width Markov random fields, \arxivalt{this has
been addressed using graph-cut algorithms 
~\cite{boykovJolly01}, which are
applicable when the MRF's energy function is submodular and limited
degree, and more general submodular function minimization in the case
of higher degree.}{the energy functions are often submodular~\cite{boykovJolly01}.} Moreover, many useful non-submodular energy functions have
also recently been phrased as forms of constrained submodular minimization~\cite{jegelka2011-nonsubmod-vision} which still have bounded
approximation guarantees. Some of these non-submodular energy functions can be modeled through Problems 1 and 2\arxivalt{,
where $f$ is still the submodular energy function to be minimized while
$g$ represents a submodular constraint.}{.} For example, in image co-segmentation
\cite{rother2006cosegmentation}, $V=V_1 \cup V_2$ can represent the
set of pixels in two images, $f$ is the submodular energy function of
the two images, while $g$ represents the similarity between the two
histograms of the hypothesized foreground regions in the two images, a
function shown to be submodular in
\cite{rother2006cosegmentation}\footnote{
\cite{rother2006cosegmentation} showed that $-g$ is supermodular.}. \looseness-1}%We can apply this to our Problem 1
%when $g = \kappa-h$ with $\kappa$ a positive constant, rendering $g$
%submodular.

Apart from the real-world applications above, both Problems 1 and 2
generalize a number of well-studied discrete optimization
problems. For example the \emph{Submodular Set Cover} problem
(henceforth SSC)~\cite{wolsey1982analysis} occurs as a special case of Problem 1,
with $f$ being modular and $g$ is submodular. Similarly the
\emph{Submodular Cost Knapsack} problem (henceforth SK)~\cite{sviridenko2004note} is a
special case of problem 2 again when $f$ is modular and $g$
submodular.  Both these problems subsume the \emph{Set
  Cover} and \emph{Max k-Cover}
problems~\cite{feige1998threshold}. When both $f$ and $g$
are modular, Problems 1 and 2 are called \emph{knapsack
  problems}~\cite{kellerer2004knapsack}.\arxiv{Furthermore, Problem 1 also subsumes the cardinality constrained submodular minimization problem~\cite{svitkina2008submodular} and more generally the problem of minimizing a submodular function subject to a knapsack constraints. It also subsumes the problem of minimizing a submodular function subject to a matroid span constraint (by setting $g$ as the rank function of the matroid, and $c$ as the rank of the matroid). This in turn subsumes the minimum submodular spanning tree problem~\cite{goel2009approximability}, when $f$ is monotone submodular.
}\looseness-1

%\arxiv{Similarly Problem 1 subsumes the problem of minimizing a submodular function subject to a matroid span (or equivalently base constraint if $f$ is monotone), which then subsumes the minimum submodular spanning tree, perfect matching and cardinality lower bound~\cite{goel2009approximability, svitkina2008submodular}.} 

\arxiv{\subsection{Our Contributions}}

The following are some of our contributions. We show that Problems 1 and 2 are intimately connected,
in that any approximation algorithm for either problem can be used to
provide guarantees for the other problem as well.  We then provide a framework of
combinatorial algorithms based on optimizing, sometimes iteratively,
subproblems that are easy to solve. These subproblems
are obtained by computing either upper or lower bound approximations
of the cost functions or constraining functions. We also show that
many combinatorial algorithms like the greedy algorithm for SK~\cite{sviridenko2004note} and SSC~\cite{wolsey1982analysis} \arxiv{(both of which seemingly use different
techniques) }also belong to this framework and
provide the first 
%\JTR{Double check the ref \cite{krause06near} as they have something very similar.}\RTJ{They do not show any bicriterion approximation for SSC.}
constant-factor bi-criterion approximation algorithm for
SSC~\cite{wolsey1982analysis} and hence the general set cover
problem~\cite{feige1998threshold}.  We then show how with suitable
choices of approximate functions, we can obtain a number of bounded
approximation guarantees and show the hardness for Problems 1 and 2,
which in fact match some of our approximation guarantees\arxiv{ up to
  $\log$-factors}. Our guarantees and hardness results depend on the
\emph{curvature} of the submodular
functions~\cite{conforti1984submodular}.  We observe a strong
asymmetry in the results that the factors change polynomially based on
the curvature of $f$ but only by a constant-factor with the curvature
of $g$, hence making the SK and SSC much easier compared to SCSK and
SCSC. \arxiv{Finally we empirically evaluate the performance of our
algorithms showing that the typical case behavior is much better than
these worst case bounds.}\looseness-1
%\subsection{Main Ideas}
%The main idea of this paper, is to provide a framework of algorithms for Problems 1 and 2. These algorithms are based on reducing the problem into simpler problems, either by replacing $f$ or $g$ with surrogate submodular functions. Further many of these algorithms are iterative and depend on iteratively solving these simpler problems. The surrogate functions are either lower/ upper bounds of the functions or approximations of the functions. 

\JTR{It's also important here to summarize the result showing how the curvature of $f$ and $g$
has a non-symmetric effect on the hardness of the problem, in that in one case
it only changes a constant, but the other case it changes a function of $n$. 
This is the stuff that is discussed briefly at the end of \S\ref{sec:hardness} but
I think this is important and actually should also be added to the abstract (not yet done).
Could
you please add a bit of text here and in the abstract about this?}\RTJ{Added something above. Pl check.}

%%%%%%%%%%%%%%%%%%%%%%%%%%%%%%%%%%%%%%%%%%%%%%%%%%%%%%%%%%%%%%%%%%%%%%%
\section{Background and Main Ideas}
\label{background}

We first introduce several key concepts used throughout the
paper. \arxivalt{}{This paper includes only the main results and we
  defer all the proofs and additional discussions to the extended
  version~\cite{nipsextendedvsubcons}. }%
Given a submodular function $f$, we define the total curvature,
$\curvf{f}$\arxiv{, and the curvature of $f$ with respect to a set $X$,
$\curvf{f}(X)$,} as\footnote{We can assume, w.l.o.g that $f(j) > 0, g(j) > 0, \forall j \in V$\arxiv{, since if for any $j \in V$, $f(j) = 0$, it follows from submodularity and monotonicity that $f(j | X) = 0, \forall X \subseteq V$. Hence we can effectively remove that element $j$ from the ground set.}}:
\arxivalt{\begin{align}
\curvf{f} = 1 - \min_{j \in V} \frac{f(j | V \backslash j)}{f(j)},\;\;\text{ and }\;\; \curvf{f}(X) = 1 - \min\{\min_{j \in X} \frac{f(j | X \backslash j)}{f(j)}, \min_{j \notin X} \frac{f(j | X)}{f(j)}\}.
\end{align}}{$\curvf{f} = 1 - \min_{j \in V} \frac{f(j | V \backslash j)}{f(j)}$~\cite{conforti1984submodular\arxiv{, vondrak2010submodularity}}.}
%\JTR{is this right, swap numerator and denominator?}
\arxiv{The total curvature $\curvf{f}$ is then
$\curvf{f}(V)$~\cite{conforti1984submodular, vondrak2010submodularity}.}
Intuitively, the curvature $0 \leq \curvf{f} \leq 1$ measures the
distance of $f$ from modularity and $\curvf{f} = 0$ if and only if $f$
is modular (or additive, i.e., $f(X) = \sum_{j \in X} f(j)$). \arxiv{Totally normalized \cite{cun82} and
saturated functions like matroid rank have a curvature
$\curvf{f} = 1$. }A number of approximation guarantees in the context of submodular optimization have been refined
via the curvature of the submodular function~\cite{conforti1984submodular, rkiyersemiframework2013, curvaturemin}. \arxiv{For example, when
maximizing a monotone submodular function under cardinality upper
bound constraints, the bound of $1 - e^{-1}$ has been refined to
$\frac{1 - e^{-\curvf{f}}}{\curvf{f}}$
~\cite{conforti1984submodular}. Similar bounds have also been shown in
the context of constrained submodular
minimization~\cite{rkiyersemiframework2013, curvaturemin, iyermirrordescent}.} In this paper, we shall
witness the role of curvature also in determining the
approximations and the hardness of problems 1 and
2. %Further many submodular functions used in practice are not too curved (i.e., they have $c_f < 1$) and our bounds suggest improved performance for these class of functions.
\narxiv{\begin{wrapfigure}[9]{r}{0.5\textwidth}
\vspace{-4.1ex}
\begin{minipage}{0.5\textwidth}}
\begin{algorithm}[H]
\caption{General algorithmic framework to address both Problems 1 and 2}
\begin{algorithmic}[1]
\FOR{$t = 1, 2, \cdots, T$}
\STATE Choose surrogate functions $\hat{f_t}$ and $\hat{g_t}$ 
for $f$ and $g$ respectively, tight at $X^{t-1}$.
\STATE Obtain $X^t$ as the optimizer of Problem 1 or 2 with
$\hat{f_t}$ and $\hat{g_t}$ instead of $f$ and $g$.
\ENDFOR
\end{algorithmic}
\label{alg:framework}
\end{algorithm}
\narxiv{\end{minipage}
\end{wrapfigure}}
The main idea of this paper is a framework of algorithms based on
choosing appropriate surrogate functions for $f$ and $g$ to optimize
over. This framework is represented in Algorithm~\ref{alg:framework}. We
would like to choose surrogate functions $\hat{f_t}$ and $\hat{g_t}$
such that using them, Problems 1 and 2 become easier. If the algorithm
is just single stage (not iterative), we represent the surrogates as
$\hat{f}$ and $\hat{g}$. The surrogate functions we consider in this
paper are in the forms of bounds (upper or lower) and
approximations. \arxiv{Our algorithms using upper and lower bounds are
  analogous to the majorization/ minimization algorithms proposed
  in~\cite{rkiyersemiframework2013, iyermirrordescent}, where we provided a unified
  framework of fast algorithms for submodular optimization. We show there in
  that this framework subsumes a large class of known combinatorial
  algorithms and also providing a generic recipe for different forms
  of submodular function optimization. We extend these ideas to the
  more general context of problems 1 and 2, to obtain a fast family of
  algorithms. The other type of surrogate functions we consider are
  those obtained from other approximations of the functions. One such
  classical approximation is the ellipsoidal
  approximations~\cite{goemans2009approximating}. While computing
  this approximation is time consuming, it turns out to provide the
  tightest theoretical guarantees.}
  
\textbf{Modular lower bounds: }
Akin to convex functions, submodular functions have tight modular lower bounds. These bounds are related to the subdifferential $\partial_f(Y)$ of the submodular set function $f$ at a set $Y \subseteq V$\arxivalt{, which is defined 
\cite{fujishige2005submodular}
as:
\begin{align}
\partial_f(Y) = \{y \in \mathbb{R}^n: f(X) - y(X) \geq f(Y) - y(Y)\;\text{for all } X \subseteq V\}
\end{align}
For a vector $x \in \mathbb{R}^V$ and $X \subseteq V$, we write $x(X)
= \sum_{j \in X} x(j)$.}{~\cite{fujishige2005submodular}.}  Denote a
subgradient at $Y$ by $h_Y \in \partial_f(Y)$. The extreme points of
$\partial_f(Y)$ may be computed via a greedy algorithm: Let $\perm$
be a permutation of $V$ that assigns the elements in $Y$ to the first
$|Y|$ positions ($\perm(i) \in Y$ if and only if $i \leq  |Y|$). Each
such permutation defines a chain with elements $S_0^\perm =
\emptyset$, $S^{\perm}_i = \{ \perm(1), \perm(2), \dots, \perm(i)
\}$ and $S^{\perm}_{|Y|} = Y$. This chain defines an extreme point
$h^{\perm}_Y$ of $\partial_f(Y)$ with entries \arxivalt{\begin{align}
    h^{\perm}_Y(\perm(i)) = f(S^{\perm}_i) - f(S^{\perm}_{i-1}).
  \end{align}}{$h^{\perm}_Y(\perm(i)) = f(S^{\perm}_i) -
  f(S^{\perm}_{i-1})$.}  Defined as above, $h^{\perm}_Y$
forms a lower bound of $f$, tight at $Y$ --- i.e.,
$h^{\perm}_Y(X) = \sum_{j \in X} h^{\perm}_Y(j) \leq f(X), \forall X
\subseteq V$ and $h^{\perm}_Y(Y) = f(Y)$.\looseness-1

\textbf{Modular upper bounds:}
We can also define superdifferentials $\partial^f(Y)$ of a submodular
function 
\cite{jegelka2011-nonsubmod-vision,rkiyersubmodBregman2012}
at
$Y$\arxivalt{: %Intuitively, it consists of tight linear upper bounds:
\begin{align}\label{supdiff-def}
\partial^f(Y) = \{y \in \mathbb{R}^n: f(X) - y(X) \leq f(Y) - y(Y);\text{for all } X \subseteq V\}
\end{align}}{.}
It is possible, moreover, to provide specific supergradients~\cite{rkiyersubmodBregman2012, rkiyersemiframework2013, iyermirrordescent} that define the following two modular upper bounds (when referring either one, we use $m^f_X$):
\notarxiv{\small}
\begin{align*}
m^f_{X, 1}(Y) \triangleq f(X) - \!\!\!\! \sum_{j \in X \backslash Y } f(j| X \backslash j) + \!\!\!\! \sum_{j \in Y \backslash X} f(j| \emptyset)\scalebox{1.3}{,}\;\;\;\arxiv{\\}
m^f_{X, 2}(Y) \triangleq f(X) - \!\!\! \sum_{j \in X \backslash Y } f(j| V \backslash j) + \!\!\!\! \sum_{j \in Y \backslash X} f(j| X). \nonumber
\end{align*}
\normalsize
Then $m^f_{X, 1}(Y) \geq f(Y)$ and $m^f_{X, 2}(Y) \geq f(Y), \forall Y \subseteq V$ and $m^f_{X, 1}(X) = m^f_{X, 2}(X) = f(X)$.

%\arxiv{\paragraph{DS optimization: } A problem very related to problems 1 and 2, is the problem of minimizing the difference between two submodular functions -- given submodular functions $f$ and $g$, $\min_X f(X) - g(X)$. When the constraints in both problems 1 and 2 are added to the objective as penalties, we essentially get back DS optimization. Unfortunately this problem is both NP hard and inapproximable~\cite{rkiyeruai2012}.\looseness-1}

\textbf{MM algorithms using upper/lower bounds: } Using the modular upper
and lower bounds above in Algorithm~\ref{alg:framework}, provide a class of Majorization-Minimization
(MM) algorithms, akin to the algorithms proposed
in~\cite{rkiyersemiframework2013, iyermirrordescent} for submodular optimization and
in~\cite{narasimhanbilmes,rkiyeruai2012} for DS optimization (Problem
0 above).  An appropriate choice of the bounds ensures that the
algorithm always improves the objective values for Problems 1 and
2. In particular, choosing $\hat{f_t}$ as a modular upper bound of $f$
tight at $X^t$, or $\hat{g_t}$ as a modular lower bound of $g$ tight
at $X^t$, or both, ensures that the objective value of Problems 1 and
2 always improves at every iteration as long as the corresponding
surrogate problem can be solved exactly. Unfortunately, Problems 1 and
2 are NP-hard even if $f$ or $g$ (or both) are modular~\cite{feige1998threshold}, \JTR{important:
  should cite where this is shown}\RTJ{done} and therefore the surrogate
problems themselves cannot be solved exactly. Fortunately, the
surrogate problems are often much easier than the original ones and
can admit $\log$ or constant-factor guarantees.  In practice,
moreover, these factors are almost $1$. 
\arxivalt{In order to guarantee improvement from a theoretical stand-point however, the iterative schemes can be slightly modified using the following trick.  Notice that the only case when the true valuation of
$X^t$ is better than $X^{t+1}$ is when the surrogate valuation of
$X^t$ is better than $X^{t+1}$ (since $X^{t+1}$ is only near-optimal
and not optimal). In such case, we terminate the algorithm at $X^t$.}{Furthermore, with a simple modification of the iterative procedure of Algorithm~\ref{alg:framework}, we can guarantee improvement at every iteration~\cite{nipsextendedvsubcons}.}
What is also fortunate and perhaps surprising, as we show in this
paper below, is that unlike the case of DS optimization (where the problem
is inapproximable in general \cite{rkiyeruai2012}), the constrained
forms of optimization (Problems 1 and 2) do have approximation
guarantees. 
\JTR{I reworded this but double check.}\RTJ{I think that’s right}

% \JTR{the below seems redundant with the above, so I'm removin it.}
% These algorithms are in principle very akin to the framework of
% algorithms for DS optimization, proposed in~\cite{rkiyeruai2012}.  The
% algorithms in~\cite{rkiyeruai2012} also use the same ideas of
% replacing $f$ and $g$ by modular upper and lower bounds respectively
% (not necessarily simultaneously). A very similar idea was used
% in~\cite{rkiyersemiframework2013} in the context of submodular
% optimization. The authors show that these modular bounds provide a
% unifying framework for various forms of submodular maximization and
% minimization.

% The analogous algorithms in~\cite{rkiyeruai2012} however do not have
% any approximation guarantees.  Similarly in the context of
% submodular optimization,~\cite{rkiyersemiframework2013} showed how
% this idea provides class of fast algorithms for constrained
% submodular minimization and a unifying framework of combinatorial
% algorithms for submodular maximization, in fact subsuming many
% commonly used variants of greedy and local search techniques.  This
% is not surprising given the result in~\cite{rkiyeruai2012}, where
% the authors show the inapproximability of DS optimization.
% \JTR{what is not surprising? It sounds like you're saying that
% ``since ~\cite{rkiyeruai2012} shows not X, it is not surprising that
% not X'' ???}

\textbf{Ellipsoidal Approximation: }
We also consider ellipsoidal approximations (EA) of $f$. The main
result of Goemans et.\  al~\cite{goemans2009approximating} is to
provide an algorithm based on approximating the submodular polyhedron
by an ellipsoid. They show that for any polymatroid function $f$,
one can compute an approximation of the form $\sqrt{w^f(X)}$ for a
certain modular weight vector $w^f \in \mathbb R^V$, such that $\sqrt{w^f(X)} \leq f(X) \leq
O(\sqrt{n}\log{n}) \sqrt{w^f(X)}, \forall X \subseteq V$.
A simple trick then provides a curvature-dependent approximation~\cite{curvaturemin} ---
we define the $\curvf{f}$-\emph{\truncated{}} version of $f$ as
follows: \arxivalt{\begin{align}
\label{eq:defineg}
f^{\kappa}(X) \triangleq 
\frac{f(X) - {(1 - \curvf{f})} \sum_{j \in X} f(j)}
{\curvf{f}}.
\end{align}}{$f^{\kappa}(X) \triangleq \bigl[f(X) - {(1 - \curvf{f})} \sum_{j \in X} f(j)\bigr]/ \curvf{f}$. Then, the submodular function $f^{\text{ea}}(X) = \curvf{f} \sqrt{w^{f^{\kappa}}(X)} + (1 -
  \curvf{f})\sum_{j \in X} f(j)$ satisfies~\cite{curvaturemin}:
\begin{align}\label{maineq}
f^{\text{ea}}(X) 
\leq f(X) 
\leq O\left(\frac{\sqrt{n} \log{n}}{1 + (\sqrt{n} \log{n} - 1)(1 - \curvf{f})}\right) f^{\text{ea}}(X), \forall X \subseteq V
\end{align}}
\arxiv{The function $f^{\kappa}$ essentially contains the zero curvature component of the submodular function $f$ and the modular upper bound $\sum_{j \in X} f(j)$ contains all the linearity. The main idea is to then approximate only the polymatroidal part and retain the linear component. This simple trick improves many of the approximation bounds.
\JTR{Add something more here about the properties of $f^\kappa$, and why
it is called curve-normalized.}\RTJ{done}
We then have the following lemma.

\begin{lemma} \cite{curvaturemin}\label{cor:learn}
  Given a polymatroid function $f$ with a curvature $\curvf{f} < 1$,
  the submodular function $f^{\text{ea}}(X) = \curvf{f} \sqrt{w^{f^{\kappa}}(X)} + (1 -
  \curvf{f})\sum_{j \in X} f(j)$ satisfies:
\begin{align}\label{maineq}
f^{\text{ea}}(X) 
\leq f(X) 
\leq O\left(\frac{\sqrt{n} \log{n}}{1 + (\sqrt{n} \log{n} - 1)(1 - \curvf{f})}\right) f^{\text{ea}}(X), \forall X \subseteq V
\end{align}
\end{lemma}} 
$f^{\text{ea}}$ is multiplicatively bounded by $f$ by a
factor depending on $\sqrt{n}$ and the curvature. \arxiv{The
  dependence on the curvature is evident from the fact that when
  $\curvf{f} = 0$, we get a bound of $O(1)$, which is not surprising
  since a modular $f$ is exactly represented as $f(X) = \sum_{j \in
    X} f(j)$. }We shall use the result above in providing
approximation bounds for Problems 1 and 2. In particular, the
surrogate functions $\hat{f}$ or $\hat{g}$ in
Algorithm~\ref{alg:framework}
can be the ellipsoidal
approximations above, and the multiplicative bounds transform into
approximation guarantees for these problems.

\section{Relation between SCSC and SCSK}
In this section, we show a precise relationship between Problems 1 and
2. From the formulation of Problems 1 and 2, it
is clear that these problems are duals of each other. Indeed, in
this section we show that
the problems are polynomially transformable into each other.

\begin{minipage}{\textwidth}
\notarxiv{\vspace{-1.1ex}}
  \centering
  \begin{minipage}{.42\textwidth}
    \centering
    \begin{algorithm}[H]
    \caption{Approx.\ algorithm for SCSK using an approximation algorithm for SCSC\arxiv{ using Linear search}.} 
    \begin{algorithmic}[1]
      \STATE \textbf{Input:} An SCSK instance with budget $b$,
      an $[\sigma, \rho]$ approx.\ algo.\ for SCSC, \& $\epsilon > 0$.
    \STATE \textbf{Output: } $[(1 - \epsilon) \rho, \sigma]$ approx. for SCSK.
\STATE $c \leftarrow g(V), \hat{X_c} \leftarrow V$.
\WHILE{$f(\hat{X_c}) > \sigma b$}
\STATE $c \leftarrow (1 - \epsilon) c$
\STATE $\hat{X_c} \leftarrow [\sigma, \rho]$ approx. for SCSC using $c$.
    \ENDWHILE
    \arxiv{\STATE Return $\hat{X}_c$}
    \end{algorithmic}
    \label{alg:alg1}
    \end{algorithm}
  \end{minipage}
~~~
  \begin{minipage}{.42\textwidth}
    \centering
    \begin{algorithm}[H]
     \caption{Approx.\ algorithm for SCSC using an approximation algorithm for SCSK\arxiv{ using Linear search}.} 
    \begin{algorithmic}[1]
      \STATE \textbf{Input:} An SCSC instance with cover $c$, an
      $[\rho, \sigma]$ approx.\ algo.\ for SCSK, \& $\epsilon > 0$.
    \STATE \textbf{Output: } $[ (1 + \epsilon)\sigma, \rho]$ approx. for SCSC.
\STATE $b \leftarrow \argmin_j f(j), \hat{X_b} \leftarrow \emptyset$.
\WHILE{$g(\hat{X_b}) < \rho c$}
\STATE $b \leftarrow (1 + \epsilon) b$
\STATE $\hat{X_b} \leftarrow [\rho, \sigma]$ approx. for SCSK using $b$.
    \ENDWHILE
    \arxiv{\STATE Return $\hat{X}_b$.}
    \end{algorithmic}
    \label{alg:alg2}
    \end{algorithm}
  \end{minipage}
  \arxiv{\captionof{figure}{Bicriterion transformation algorithms for SCSC and SCSK using Linear search.}}
  \label{fig:twoalg}
\end{minipage}
\normalsize
\arxiv{\vspace{1.1ex}}

We first introduce the notion of
bicriteria algorithms. An algorithm is a $[\sigma, \rho]$ bi-criterion
algorithm for Problem 1 if it is guaranteed to obtain a set $X$ such
that $f(X) \leq \sigma f(X^*)$ (approximate optimality) and $g(X) \geq
c^{\prime} = \rho c$ (approximate feasibility), where $X^*$ is an optimizer of
Problem 1. Similarly, an algorithm is a $[\rho,\sigma]$ bi-criterion algorithm for Problem 2 if it is guaranteed to
obtain a set $X$ such that $g(X) \geq \rho g(X^*)$ and $f(X) \leq b^{\prime} = 
\sigma b$, where $X^*$ is the optimizer of Problem 2. In a
bi-criterion algorithm for Problems 1 and 2, typically $\sigma \geq 1$
and $\rho \leq 1$. \arxiv{We call these type of approximation algorithms, bi-criterion approximation algorithms of type 1. 

We can also view the bi-criterion approximations from another angle. We can say that for Problem 1, $\hat{X}$ is a feasible solution (i.e., it satisfies the constraint), and is a $[\sigma, \rho]$ bi-criterion approximation, if $f(\hat{X}) \leq \sigma f(\hat{X}^*)$, where $\hat{X}^*$ is the optimal solution to the problem $\min\{f(X) | g(X) \geq c/\rho\}$. Similarly for Problem 2, we can say that $\hat{X}$ is a feasible solution (i.e., it satisfies the constraint), and is a $[\rho, \sigma]$ bi-criterion approximation, if $g(\hat{X}) \geq \rho g(\hat{X}^*)$, where $\hat{X}^*$ is the optimal solution to the problem $\max\{g(X) | f(X) \leq b \sigma\}$. We call these the bi-criterion approximation algorithms of type 2. 

It is easy to see that these algorithms can easily be transformed into each other. For example, for problem 1, a bi-criterion algorithm of type-I can obtain a guarantee of type-II if we run it till a covering constraint of $c^{\prime}/\rho$ (where $c^{\prime}$ in this case, is the approximate covering constraint which the type-I algorithm needs to satisfy -- note that it need not be the actual covering constraint of the problem). Similarly an algorithm of type-II can obtain a guarantee of type-I if run till a covering constraint of $\rho c$ (in this case, $c$ is the actual 'covering' constraint, since a type-II approximate algorithm provides a feasible set). We can similarly transform these guarantees for problem 2. In particular, a bi-criterion algorithm of type-I can be used to obtain a guarantee of type-II if we run it till a budget of $b^{\prime} \sigma$ (again, here $b^{\prime}$ is the approximate budget of the type-I algorithm). Similarly an algorithm of type-II can obtain a guarantee of type-I if run till a covering constraint of $b/\sigma$ (in this case, $b$ is the budget of the original problem).

Though both type-I and type-II guarantees are easily transformable into each other, through the rest of this paper whenever we refer to bi-criterion approximations, we shall consider only the type-I approximations.
}A {\em non-bicriterion} algorithm for Problem 1 is
when $\rho = 1$ and a {\em non-bicriterion} algorithm for Problem 2 is when
$\sigma = 1$.
Algorithms~\ref{alg:alg1} and \ref{alg:alg2} provide the schematics
for using an approximation algorithm for one of the problems
for solving the other.\narxiv{\looseness-1}\arxiv{ 

The main idea of
Algorithm~\ref{alg:alg1} is to start with the ground set $V$ and
reduce the value of $c$ (which governs SCSC), until the valuation of
$f$ just falls below $\sigma b$. At that point, we are guaranteed to
get a $((1 - \epsilon) \rho, \sigma)$ solution for SCSK. Similarly in
Algorithm~\ref{alg:alg2}, we increase the value of $b$ starting at the
empty set, until the valuation at $g$ falls above $\rho c$. At this
point we are guaranteed a $(( 1 + \epsilon) \sigma, \rho)$ solution for SCSC. In order to avoid degeneracies, we assume that $f(V) \geq b \geq \min_j f(j)$ and $g(V) \geq c \geq \min_j g(j)$, else the solution to the problem is trivial.}
\begin{theorem}
\label{thm1}
  Algorithm~\ref{alg:alg1} is guaranteed to find a set $\hat{X_c}$
  which is a $[(1 - \epsilon) \rho, \sigma]$ approximation of SCSK in at most 
  $\log_{1/(1 - \epsilon)} [g(V)/\min_j g(j)]$ calls to the $[\sigma, \rho]$ approximate
  algorithm for SCSC. Similarly, Algorithm~\ref{alg:alg2} is
  guaranteed to find a set $\hat{X_b}$ which is a $[(1 + \epsilon)
  \sigma, \rho]$ approximation of SCSC in $\log_{1 + \epsilon} [f(V)/\min_j f(j)]$ calls
  to a $[\rho, \sigma]$ approximate algorithm for SCSK.
\end{theorem}
\arxiv{
\begin{proof}
  We start by proving the first part, for Algorithm~\ref{alg:alg1}. Notice that
  Algorithm~\ref{alg:alg1} converges when $f(\hat{X_c})$ just falls
  below $\sigma b$. Hence $f(\hat{X_c}) \leq \sigma b$ (is approximately feasible) and at
  the previous step $c^{\prime} = c/(1 - \epsilon)$, we have that
  $f(\hat{X_{c^{\prime}}}) > \sigma b$. Denoting
  $X^*_{c^{\prime}}$ as the optimal solution for SCSC at $c^{\prime}$,
  we have that $f(X^*_{c^{\prime}}) > b$ (a fact which follows from
  the observation that $\hat{X_c}$ is a $[\sigma, \rho]$ approximation
  of SCSC at $c$). Hence if $X^*$ is the optimal solution of SCSK, it
  follows that $g(X^*) < c^{\prime}$. The reason for this is that,
  suppose, $g(X^*) \geq c^{\prime}$. Then it follows that $X^*$ is a
  feasible solution for SCSC at $c^{\prime}$ and hence $f(X^*) \geq
  f(X^*_{c^{\prime}}) > b$. This contradicts the fact that $X^*$ is an
  optimal solution for SCSK (since it is then not even feasible).

  Next, notice that $\hat{X_c}$ satisfies that $g(\hat{X_c}) \geq \rho
  c$, using the fact that $\hat{X_c}$ is obtained from a $(\sigma,
  \rho)$ bi-criterion algorithm for SCSC. Hence,
\begin{align}
g(\hat{X_c}) 
\geq \rho c 
= \rho (1 - \epsilon) c^{\prime} 
> \rho (1 - \epsilon) g(X^*) 
\end{align}
Hence the Algorithm~\ref{alg:alg1} is a $((1 - \epsilon) \rho,
\sigma)$ approximation for SCSK. 
In order to show the converge rate, notice that $c \geq \min_j g(j) > 0$. Since $\sigma \geq 1$ and $b \geq \min_j f(j)$, we can guarantee that this algorithm will stop before $c$ reaches $\min_j g(j)$. The reason is that, when $c = \min_j g(j)$, the minimum value of $f(X)$ such that $g(X) \geq c$ is $\min_j f(j)$, which is smaller than $b$. Moreover, since $\sigma > 1$, it implies that the algorithm would have terminated before this point.

%The reason is that when $c \leq \min_j g(j)$, $\hat{X}_c = \emptyset$. Hence the algorithm starts at $c = g(V)$, and will converge in $\log_{1/(1 - \epsilon)} [g(V)/\min_j g(j)]$ iterations. Note that, for small $\epsilon$, this is roughly $\frac{\log[g(V)/\min_j g(j)]}{\epsilon} = O(1/\epsilon)$.

%
%%%\JTR{Is this right, or is it $\log_{1/(1-\epsilon)} (g(V)/g(X))$ steps,
%where $X$ is the solution of the algorithm, and where we could have $g(X) < 1$??}\RTJ{You’re right. I changed it above.}
%\JTR{Note: once we have $c$ and $c'$, I think we can do
%  another bisection search between these two values that, in log time,
%  will remove the $(1-\epsilon)$ factor from the approximation factor.
%}%\RTJ{I think the bisection method works only for integral $f$ and $g$ right? If not then, we can really do bisection to get the exact approximation right? This is slightly confusing to me, hence I kept the one it is currently.}
%
The proof for the second part of the statement, for
Algorithm~\ref{alg:alg2}, is omitted since it is shown using a
symmetric argument.  
\end{proof}
}%

\arxiv{\noindent\begin{minipage}{\textwidth}
  \centering
  \begin{minipage}{.42\textwidth}
    \centering
    \begin{algorithm}[H]
    \caption{Approx.\ algorithm for SCSK using an approximation algorithm for SCSC\arxiv{ using Binary search}.} 
    \begin{algorithmic}[1]
      \STATE \textbf{Input:} An SCSK instance with budget $b$,
      an $[\sigma, \rho]$ approx.\ algo.\ for SCSC, \& $\epsilon > 0$.
    \STATE \textbf{Output: } $[(1 - \epsilon) \rho, \sigma]$ approx. for SCSK.
    \STATE $c_{\min} \leftarrow \min_j g(j), c_{\max} \leftarrow g(V)$
\WHILE{$c_{\max} - c_{\min} \geq \epsilon c_{\max}$ }
\STATE $c \leftarrow [c_{\max} + c_{\min}]/2$
\STATE $\hat{X_c} \leftarrow [\sigma, \rho]$ approx. for SCSC using $c$.
\IF{$f(\hat{X_c}) > \sigma b$}
\STATE $c_{\max} \leftarrow c$
\ELSE
\STATE $c_{\min} \leftarrow c$
\ENDIF
    \ENDWHILE
\STATE Return $\hat{X}_{c_{\min}}$
    \end{algorithmic}
    \arxivalt{\label{alg:alg3}}{\label{alg:alg1}}
    \end{algorithm}
  \end{minipage}
~~~
  \begin{minipage}{.42\textwidth}
    \centering
    \begin{algorithm}[H]
     \caption{Approx.\ algorithm for SCSC using an approximation algorithm for SCSK\arxiv{ using Binary search}.} 
    \begin{algorithmic}[1]
      \STATE \textbf{Input:} An SCSC instance with cover $c$, an
      $[\rho, \sigma]$ approx.\ algo.\ for SCSK, \& $\epsilon > 0$.
    \STATE \textbf{Output: } $[ (1 + \epsilon)\sigma, \rho]$ approx. for SCSC.
\STATE $b_{\min} \leftarrow \argmin_j f(j), b_{\max} \leftarrow f(V)$.
\WHILE{$b_{\max} - b_{\min} \geq \epsilon b_{\min}$ }
\STATE $b \leftarrow [b_{\max} + b_{\min}]/2$
\STATE $\hat{X_b} \leftarrow [\rho, \sigma]$ approx. for SCSK using $b$
\IF{$g(\hat{X_b}) < \rho c$}
\STATE $b_{\min} \leftarrow b$
\ELSE
\STATE $b_{\max} \leftarrow b$
\ENDIF
    \ENDWHILE
    \STATE Return $\hat{X}_{b_{\max}}$
    \end{algorithmic}
    \arxivalt{\label{alg:alg4}}{\label{alg:alg2}}
    \end{algorithm}
  \end{minipage}
  \arxiv{\captionof{figure}{Bicriterion transformation algorithms for SCSC and SCSK using Binary search.}}
  \label{fig:twoalg}
\end{minipage}}
\arxiv{\vspace{1.1ex}}

%\notarxiv{
%\begin{theorem}
%\label{thm1}
%  Algorithm~\ref{alg:alg1} is guaranteed to find a set $\hat{X_c}$
%  which is a $[(1 - \epsilon) \rho, \sigma]$ approximation of SCSK in at most 
%  $\log_2 \frac{[g(V)/\min_j g(j)]}{\epsilon}$ calls to the $[\sigma, \rho]$ approximate
%  algorithm for SCSC. Similarly, Algorithm~\ref{alg:alg2} is
%  guaranteed to find a set $\hat{X_b}$ which is a $[(1 + \epsilon)
%  \sigma, \rho]$ approximation of SCSC in $\log_2 \frac{[f(V)/\min_j f(j)]}{\epsilon}$ calls
%  to a $[\rho, \sigma]$ approximate algorithm for SCSK. When $f$ and $g$ are integral, moreover, Algorithm~\ref{alg:alg1} obtains a $[\rho, \sigma]$ bi-criterion approximation for SCSK in $\log_2 g(V)$ iterations, and Algorithm~\ref{alg:alg2} obtains a $[\sigma, \rho]$ bi-criterion approximation for SCSC in $\log_2 f(V)$ iterations.
%\end{theorem}
%}
Theorem~\ref{thm1} implies that the complexity of Problems 1 and 2 are
identical, and a solution to one of them provides a solution to the
other. Furthermore, as expected, the hardness of Problems 1 and 2 are
also almost identical. When $f$ and $g$ are polymatroid functions,
moreover, we can provide bounded approximation guarantees for both
problems, as shown in the next section. \narxiv{Alternatively we can also do a binary search instead of a linear search to transform Problems 1 and 2. This essentially turns the factor of $O(1/\epsilon)$ into $O(\log 1/\epsilon)$. Due to lack of space, we defer this discussion to the extended version~\cite{nipsextendedvsubcons}.\looseness-1}

\arxiv{Alternatively we can also do a binary search instead of a linear search to transform Problems 1 and 2. This essentially turns the factor of $O(1/\epsilon)$ into $O(\log 1/\epsilon)$. Algorithms~\ref{alg:alg3} and \ref{alg:alg4} show the transformation algorithms using binary search. While the binary search also ensures the same performance guarantees, it does so more efficiently.

\begin{theorem}
\label{thm2}
  Algorithm~\ref{alg:alg3} is guaranteed to find a set $\hat{X_c}$
  which is a $[(1 - \epsilon) \rho, \sigma]$ approximation of SCSK in at most 
  $\log_2 \frac{[g(V)/\min_j g(j)]}{\epsilon}$ calls to the $[\sigma, \rho]$ approximate
  algorithm for SCSC. Similarly, Algorithm~\ref{alg:alg4} is
  guaranteed to find a set $\hat{X_b}$ which is a $[(1 + \epsilon)
  \sigma, \rho]$ approximation of SCSC in $\log_2 \frac{[f(V)/\min_j f(j)]}{\epsilon}$ calls
  to a $[\rho, \sigma]$ approximate algorithm for SCSK. When $f$ and $g$ are integral, moreover, Algorithm~\ref{alg:alg3} obtains a $[\rho, \sigma]$ bi-criterion approximate solution for SCSK in $\log_2 g(V)$ iterations, and similarly Algorithm~\ref{alg:alg4} obtains a $[\sigma, \rho]$ bi-criterion approximate solution for SCSC in $\log_2 g(V)$ iterations.
\end{theorem}
\begin{proof}
To show this theorem, we use the result from Theorem~\ref{thm1}. Let $c = c_{\min}$ and $c^{\prime} = c_{\max}$. An important observation is that throughout the algorithm, the values of $c_{\min}$ satisfy $f(\hat{X}_{c_{\min}}) \leq \sigma b$ and $f(\hat{X}_{c_{\max}}) > \sigma b$. Hence, $f(\hat{X_{c}}) \leq \sigma b$ and $f(\hat{X_{c^{\prime}}}) > \sigma b$. Moreover, notice that $c^{\prime}/c = c_{\max}/c_{\min} = 1/(1 - \epsilon)$. Hence using the proof of Theorem~\ref{thm1} the approximation guarantee follows. 

In order to show the complexity, notice that $c_{\max} - c_{\min}$ is decreasing throughout the algorithm. At the beginning, $c_{\max} - c_{\min} \leq g(V)$ and at convergence, $c_{\max} - c_{\min} \geq \epsilon c_{\max}/2 \geq \epsilon \min_j g(j)/2$. The bound at convergence holds since, let $c^{\prime}_{max}$ and $c^{\prime}_{min}$ be the values at the previous step. It holds that $c^{\prime}_{\max} - c^{\prime}_{\min} \geq \epsilon c^{\prime}_{\max}$. Moreover, $c_{max} - c_{min} = (c^{\prime}_{\max} - c^{\prime}_{\min})/2 \geq \epsilon c^{\prime}_{\max}/2 \geq \epsilon c_{\max}$. Hence the number of iterations is bounded by $\log_2 \frac{[2g(V)/\min_j g(j)]}{\epsilon}$. Moreover, when $f$ and $g$ are integral, the analysis is much simpler. In particular, notice that once $c_{max} - c_{min} = 1$, the algorithm will stop at the next iteration (this is because at this point, $c = (c_{max} + c_{min})/2$ is equivalent to $c = c_{max}$). Hence, the number of iterations is bounded by $\log_2 g(V)$, and we can exactly obtain a $[\rho, \sigma]$ bi-criterion approximation algorithm.

The proof for the second part of the statement, for
Algorithm~\ref{alg:alg2}, similarly follows using a symmetric argument.  

\end{proof}
 When $f$ and $g$ are integral, this removes the $\epsilon$ dependence on the factors and could potentially be much faster in practice. We also remark that a specific instance of such a transformation has been used~\cite{krause08robust}, for a specific class of functions $f$ and $g$. We shall show in section~\ref{extentions} that their algorithm is in fact a special case of Algorithm~\ref{alg:alg3} through a specific construction to convert the non-submodular problem into an instance of a submodular one.
 
  Algorithm~\ref{alg:alg1} and Algorithm~\ref{alg:alg2} indeed provide an
  interesting theoretical result connecting the complexity of Problems
  1 and 2. In the next section however, we provide distinct algorithms for each problem when $f$ and $g$ are polymatroid functions --- this turns out to
  be faster than having to resort to the iterative reductions above.}

\section{Approximation Algorithms}

We consider several algorithms for Problems 1 and 2, which can all be
characterized by the framework of
Algorithm~\ref{alg:framework}, using the surrogate functions of
the form of upper/lower bounds or approximations.\looseness-1

%%%%%%%%%%%%%%%%%%%%%%%%%%%%%%%%%%%%%%%%%%%%%%%%%%%%%%%%%%%%%%%%%
\subsection{Approximation Algorithms for SCSC}
\label{sec:appr-algor-scsc}

We first describe our approximation algorithms designed specifically
for SCSC, leaving to \S\ref{sec:appr-algor-scsk} the presentation of
our algorithms slated for SCSK. We first investigate a special case, the submodular set cover (SSC), and then provide two algorithms, one of them (ISSC) is very practical with a weaker theoretical guarantee, and another one (EASSC) which is slow but has the tightest guarantee.

\textbf{Submodular Set Cover (SSC): }
% \JTR{Note that the first thing that comes to mind as an approixmation
% algorithm is to start with problem 1 where both $f$ and $g$ are
% submodular, and then to repeatedly use the submodular set cover
% greedy algorithm using the modular upper bound for $f$ that is tight
% at the previous solution. This section sort of sounds initially
% like that is being done (but instead, what is being done
% is that $f$ is presumed modular inherently and then one is taking
% a greedy modular bound for $g$ that corresponds to the order
% chosen in the submodular set cover algorithm). I've reworded a bit
% here to this effect.}
We start by considering a classical special case of SCSC (Problem 1)
where $f$ is already a modular function and $g$ is a submodular
function. This problem occurs naturally in a number of problems
related to active/online learning~\cite{guillory2010interactive} and summarization~\cite{linbudget,lin2011-class-submod-sum}. This
problem was first investigated by Wolsey~\cite{wolsey1982analysis},
wherein he showed that a simple greedy algorithm achieves bounded (in
fact, log-factor) approximation guarantees. We show that this greedy
algorithm can naturally be viewed in the framework of our
Algorithm~\ref{alg:framework} by choosing appropriate surrogate
functions $\hat{f_t}$ and $\hat{g_t}$. The idea is to use the modular
function $f$ as its own surrogate $\hat{f_t}$ and choose the function
$\hat{g_t}$ as a modular lower bound of $g$. Akin to the framework of
algorithms in~\cite{rkiyersemiframework2013}, the crucial factor is
the choice of the lower bound (or subgradient). Define the
\emph{greedy subgradient} \arxiv{(equivalently the \emph{greedy
    permutation})} as:
\begin{align}
\label{greedyssc}
\perm(i) \in 
\argmin \left\{\frac{f(j)}{g(j | S^{\perm}_{i-1})} \ \bigg\vert \ j \notin S^{\perm}_{i-1} , 
         g(S^{\perm}_{i-1} \cup j) < c \right\}.
% ,\;\;\text{ if } g(S^{\sigma}_{i-1} \cup j) < c
\end{align}
Once we reach an $i$ where the constraint $g(S^{\perm}_{i-1} \cup j)
< c$ can no longer be satisfied by any $j \notin S^{\perm}_{i-1}$, we
choose the remaining elements for $\perm$ arbitrarily. Let the corresponding
subgradient be referred to as $h^{\perm}$.
% \JTR{I think we should make this more
%   clear (also in the other paper) what is done if the constraint is
%   not satisfied. I.e., at each of those points we choose arbitrarily
%   (we should be able to specify the algorithm still on one line to
%   save space).}  
\arxiv{Let $N$ be the minimum $i$
  such that $g(S^{\perm}_i) \geq c$ and $\theta_i = \min_{j \notin S^{\perm}_{i-1}} \frac{f(j)}{g(j |
    S^{\perm}_{i-1})}$.} Then we have the following lemma, 
which is an extension of \cite{wolsey1982analysis}\arxivalt{:}{, }
\arxivalt{\begin{lemma} \label{setcover} Choosing the surrogate
    function $\hat{f}$ as $f$ and $\hat{g}$ as $h^{\perm}$ (with
    $\perm$ defined in Eqn.~\eqref{greedyssc}) in
    Algorithm~\ref{alg:framework}, at the end of the first iteration,
    we are guaranteed to obtain a set $X^g$ such that
\begin{align}
\frac{f(X^g)}{f(X^*)} \leq 1 + \log_e \min\{\lambda_1, \lambda_2, \lambda_3\}
\end{align} 
where $\lambda_1 = \max \{\frac{1}{1 - \kappa_g(S^{\perm}_i)} | \
i: c(S^{\perm}_i) < 1\}$, $\lambda_2 = \frac{\theta_N}{\theta_1}$
and $\lambda_3 = \frac{g(V) - g(\emptyset)}{g(V) -
  g(S^{\perm}_{N-1})}$. Furthermore if $g$ is integral,
$\frac{f(X^g)}{f(X^*)} \leq H(\max_j g(j))$, where $H(d) = \sum_{i =
  1}^d \frac{1}{i}$ for a positive integer $d$.
\end{lemma}
\begin{proof}
  The permutation $\perm$ is chosen based on a greedy ordering
  associated with the submodular set cover problem, and therefore
  $h^{\perm}(S^{\perm}_N) = g(S^{\perm}_N) \geq c$ (where the
  inequality follows from the definition of $N$), and thus
  $S^{\perm}_N$ is a feasible solution in the surrogate problem
  (where $g$ is replaced by $h^\perm$). The resulting knapsack
  problem can be addressed using the greedy algorithm
  \cite{kellerer2004knapsack}, but this exactly corresponds to the
  greedy algorithm of submodular set cover \cite{wolsey1982analysis} and hence the guarantee
  follows from \JTR{I added the theorem number, but double check that
    this is really theorem 1, I think it is} Theorem 1
  in~\cite{wolsey1982analysis}\RTJ{Yup, I think that's right}.
\end{proof}
}
{% A simpler version of the Lemma for the main paper:
and which is a simpler description of the result stated formally in~\cite{nipsextendedvsubcons}.\looseness-1
\begin{lemma}
\label{setcover}
  The greedy algorithm for SSC~\cite{wolsey1982analysis} can be seen as an instance of
  Algorithm~\ref{alg:framework} by choosing the surrogate function
  $\hat{f}$ as $f$ and $\hat{g}$ as $h^{\perm}$ (with $\perm$
  defined in Eqn.~\eqref{greedyssc}).
\end{lemma}
}
%\JTR{to make the impact of the theorem a bit easier to get the gist of,
%is there an upper bound on $1 + \log_e \min\{\lambda_1, \lambda_2, \lambda_3\}$
%that is very easy to interpret? As it stands, the bound is fairly complex and hard
%to know how good it is, so an additional looser but easier to interpret bound would be great to add here.}
\arxiv{The surrogate problem in this instance is a simple knapsack problem
that can be solved nearly optimally using dynamic
programming~\cite{vazirani2004approximation}. As stated in the
proof, the greedy
algorithm for the submodular set cover
problem~\cite{wolsey1982analysis} is in fact equivalent to
using the greedy algorithm for the knapsack
problem~\cite{kellerer2004knapsack}, which is in fact suboptimal.} When
$g$ is integral, the guarantee of the greedy algorithm is $H_g
\triangleq H(\max_j g(j))$, where $H(d) = \sum_{i = 1}^d
\frac{1}{i}$~\cite{wolsey1982analysis} (henceforth we will use $H_g$ for this
quantity).  This factor is tight
up to lower-order terms~\cite{feige1998threshold}. Furthermore, since
this algorithm directly solves SSC, we call it the \emph{primal greedy}. We could also
solve SSC by looking at its \emph{dual}, which is SK~\cite{sviridenko2004note}. Although SSC does
not admit any constant-factor approximation algorithms
\cite{feige1998threshold}, we can obtain a constant-factor {\em
  bi-criterion} guarantee:\looseness-1
\begin{lemma}
\label{thm:dual_greedy}
  Using the greedy algorithm for SK~\cite{sviridenko2004note} as the approximation oracle in
  Algorithm~\ref{alg:alg2} provides a $[1 + \epsilon, 1 - e^{-1}]$
  bi-criterion approximation algorithm for SSC, for any $\epsilon > 0$.
\end{lemma}
\JTR{Note that the Sviridenko algorithm of \cite{sviridenko2004note} requires
partial enumeration which is an $O(n^3)$ step rendering this approach no longer
practical. In our 2010 
paper \cite{linbudget}
we show that w/o the $O(n^3)$ step still
has a $1-1/\sqrt{e} \approx 0.33$ algorithm which still works well in practice,
and that would also give a different bi-criterion approximation here
of $[1+\epsilon,1-e^{-1/2}]$.
Note that this was also shown in \url{http://las.ethz.ch/files/krause05note.pdf} ,
I added the bib here \cite{krause05note},
but
at the time in 2010, we weren't aware of Andreas's result (my fault :-).
}\RTJ{Added a note below.}
\arxiv{
\begin{proof}}We call this the \emph{dual greedy}. This result follows immediately from the guarantee of the submodular cost knapsack problem~\cite{sviridenko2004note} and Theorem~\ref{thm1}. \arxiv{\end{proof}} We remark that we can also use a simpler version of the greedy iteration at every iteration~\cite{linbudget, krause05note} and we obtain a guarantee of $(1 + \epsilon, 1/2(1 - e^{-1}))$. In practice, however, both these factors are almost $1$ and hence the simple variant of the greedy algorithm suffices.\looseness-1
\arxiv{

An interesting connection between the greedy algorithm and the induced orderings, allows us to further simplify this dual algorithm. A nice property of the greedy algorithm for the submodular knapsack problem is that it can be completely parameterized by the chain of sets (this holds for the greedy algorithm of~\cite{linbudget, krause05note} for knapsack constraints and the basic greedy algorithm of~\cite{nemhauser78} under cardinality constraints). In particular, having computed the greedy chain of sets, and given a value of $b$ or the budget, we can easily find the corresponding set in $O(\log n)$ time using binary search. Moreover, this also implies that we can do the transformation algorithms by just iterating through the chain of sets once. In particular, the linear search over the different values of $\epsilon$ is equivalent to the linear search over the different chain of sets. Moreover, we could also do the much faster binary search. Hence the complexity of the dual greedy algorithm is almost identical to the primal greedy one for the submodular set cover problem.

It is also important to put the bicriterion result into perspective. Notice that the bicriterion guarantee suggests that we find only an approximate feasible solution. In particular, the dual greedy algorithm provides a type-I bi-criterion approximation, -- that the solution obtained by running the algorithm with a cover constraint of $(1 - 1/e) c$ is competitive to the optimal solution with a cover constraint of $c$. However, we can also obtain a type-II bi-criterion approximation by running the dual greedy algorithm until it satisfies the cover constraint of $c$. In this case, we would obtain a feasible solution. The guarantee, however, would say that the resulting solution would be competitive to the optimal solution obtained with a cover constraint of $c/(1 - e^{-1})$. Furthermore, these factors are in practice close to $1$, and the primal and dual greedy algorithms would both perform very well empirically.

}

\textbf{Iterated Submodular Set Cover (ISSC): }
We next investigate an algorithm for the general SCSC problem when both $f$ and $g$
are submodular. The idea here is to iteratively solve the submodular
set cover problem which can be done by replacing $f$ by a modular
upper bound at every iteration. In particular, this can be seen as a
variant of Algorithm~\ref{alg:framework}, where we start with \arxivalt{$X^0 =
\emptyset$ and choose $\hat{f_t}(X) = m^f_{X^t, 2}(X)$ at every
iteration (alternatively, we can choose $\hat{f_t}(X) = m^f_{X^t,
  1}(X)$). At the first iteration with $X^0 = \emptyset$, either
variant then corresponds to the set cover problem with the simple
modular upper bound $f(X) \leq m^f_{\emptyset}(X) = \sum_{j \in X}
f(j)$ where $m^f_{X^t}$ refers to either variant.}{$X^0 =
\emptyset$ and choose $\hat{f_t}(X) = m^f_{X^t}(X)$ at every
iteration.} The surrogate
problem at each iteration becomes\arxiv{:}
\arxivalt{
\begin{alignat*}{2}
     \text{minimize }        & m^f_{X^t}(X) \\
    \text{subject to }     & g(X) \geq c.
  \end{alignat*}
}{
% \begin{align}
$\min\{ m^f_{X^t}(X) | g(X) \geq c\}$.
%\end{align}
}Hence, each iteration is an instance of SSC and can be
solved nearly optimally using the greedy algorithm.  We can continue
this algorithm for $T$ iterations or until convergence. An analysis
very similar to the ones in~\cite{rkiyeruai2012,
  rkiyersemiframework2013} will reveal polynomial time convergence.
\JTR{for completeness, in supplementary material, this should ultimately be
  included here in its current form.} Since each iteration is only the
greedy algorithm, this approach is also highly practical and
scalable. \arxiv{Since there are two approaches to solve the set cover
problem (the primal 
approach of 
Lemma~\ref{setcover}
and the dual greedy
approach of 
Lemma~\ref{thm:dual_greedy}), we have two forms of ISSC,
the \emph{primal ISSC} and the \emph{dual ISSC}. The
following shows the resulting theoretical guarantees:}\looseness-1
\begin{theorem}\label{MIguar}
  \arxiv{The primal} ISSC\arxiv{ algorithm} obtains an approximation factor of
  $\frac{K_g H_g}{1 + (K_g - 1)(1 - \curvf{f})} \leq
    \frac{n}{1 + (n - 1)(1 - \curvf{f})} H_g$ where $K_g = 1 +
  \max \{|X|: g(X) < c\}$ and $H_g$ is the approximation factor
  of the submodular set cover using $g$. \arxiv{Similarly the dual ISSC obtains
  a bi-criterion guarantee of $\left[\frac{(1 + \epsilon) K_g}{1 + (K_g -
    1)(1 - \curvf{f})}, 1 - e^{-1}\right]$.}
\end{theorem}
\JTR{again, we can get a more scalable version w/o the $O(n^3)$ step and with approx $1-e^{-1/2}$, but
in this case maybe it is no longer practical due to the need to use the
dual solver of Algorithm~\ref{alg:alg2}.}\RTJ{Yeah. I think given that we already mentioned this remark in the case of SSC, we can just keep the $1 - 1/e$ henceforth.}
\arxiv{\begin{proof}
The first part of the result follows directly from Theorem 5.4 in~\cite{rkiyersemiframework2013}. In particular, the result from~\cite{rkiyersemiframework2013} ensures a guarantee of
\begin{align}
\frac{\beta|X^*|}{1 + (|X^*| - 1)(1 - \curvf{f})} 
\end{align}
for the problem of $\min\{f(X) | X \in \mathcal C\}$ where $\beta$ is
the approximation guarantee of solving a modular function over
$\mathcal C$ where $\mathcal C$ is the feasible set. In this case,
$\beta = H_g$ and $|X^*| \leq K_g$.

When using the dual greedy approach at every iteration, we
can use a similar form of the result in the bi-criterion sense. Consider only the first iteration of this algorithm (due to the monotonicity of the algorithm, we will only improve the objective value). We are then guaranteed to obtain a set $\hat{X}$ such that (denote $X^1$ as the solution after the first iteration)
\begin{align}
f(\hat{X}) \leq f(X^1) \leq m^f_{\emptyset}(X^1) \leq (1 + \epsilon)m^f_{\emptyset}(X^*) \leq \frac{K_g (1 + \epsilon)}{1 + (K_g - 1)(1 - \curvf{f})} f(X^*)
\end{align}
The inequalities above follow from the fact that the modular upper bound $m^f_{\emptyset}(X)$ satisfies~\cite{curvaturemin}, 
\begin{align}
m^f_{\emptyset}(X) \leq f(X) \leq \frac{|X|}{1 + (|X| - 1)(1 - \curvf{f})} f(X)
\end{align} 
and the fact that $X^1$ which is the solution obtained through the dual set cover with a cover constraint $(1 - e^{-1})c$ satisfies $m^f_{\emptyset}(X^1) \leq (1 + \epsilon)m^f_{\emptyset}(X^*)$.
In the above, $X^*$ is the optimal solution to the problem $\min\{ f(X) | g(X)
\geq c\}$. We could also run the dual set cover algorithm to obtain a feasible solution (i.e., a type-II guarantee). In this case, the guarantee would compete with the optimal solution satisfying a cover constraint of $c/(1 - e^{-1})$.
\end{proof}}

From the above, it is clear that $K_g \leq n$. Notice also that $H_g$
is essentially a log-factor. We also see an interesting effect of the
curvature $\curvf{f}$ of $f$. When $f$ is modular ($\curvf{f} = 0$),
we recover the approximation guarantee of the submodular set cover
problem. Similarly, when $f$ has restricted curvature, the guarantees can be much better. \arxiv{For example, using square-root over modular function $f(X) = \sum_{i = 1}^k \sqrt{w_i(X)}$,  which is common model used in applications~\cite{rkiyeruai2012, lin2012submodularity, jegelkacvpr}, the worst case guarantee is $H_g \sqrt{K_g}$. This follows directly from the results in~\cite{curvaturemin}. } Moreover, the approximation guarantee already holds after the
first iteration, so additional iterations can only further improve the
objective.  \arxiv{

We remark here that a special case of ISSC, using only the first iteration (i.e., the simple modular upper bound of $f$) was considered in\arxivalt{~\cite{wan2010greedy,du2011minimum}}{~\cite{wan2010greedy}}. Our algorithm not only possibly improves upon theirs, but our approximation guarantee is also more explicit than theirs. In particular, they show a guarantee of $\nu_f H_g$, where $\nu_f = \min\{\frac{\sum_{i \in X} f(i)}{f(X)} | g(X) = g(V)\}$. Since this factor $\nu_f$ itself involves an optimization problem, it is not clear how to efficiently compute this factor. Moreover given a submodular function $f$, it is also not evident how good this factor is. While our guarantee is an upper bound of $\nu_f H_g$, it is much more explicit in it’s dependence on the parameters of the problem. It can also be computed efficiently and has an intuitive significance related to the curvature of the function. Furthermore, our bound is also tight since with, for example, $f(X) = \min\{|X|, 1\}$, our exactly matches the bound of~\cite{wan2010greedy, du2011minimum}. Lastly our algorithm also potentially improves upon theirs thanks to it’s iterative nature.}

\arxiv{For another variant of this algorithm, we can replace $g$ with
  its greedy modular lower bound at every iteration. Then, rather
  than solving every iteration through the greedy algorithm, we can
  solve every iteration as a knapsack problem (minimizing a modular
  function over a modular lower bound constraint)
  \cite{kellerer2004knapsack}, using say, a dynamic programming based approach. This could potentially improve over the
  greedy variant, but at a potentially higher computational
  cost.}\looseness-1

\textbf{Ellipsoidal Approximation based Submodular Set Cover (EASSC): } 
In this setting, we use the ellipsoidal approximation discussed in
\S\ref{background}. We can compute the $\curvf{f}$-\truncated{}
version of $f$ ($f^{\kappa}$, see \S\ref{background}), and then
compute its ellipsoidal approximation $\eax{f^\kappa}$.  We then
define the function $\hat{f}(X) = f^{\text{ea}}(X) = \curvf{f} \sqrt{w^{f^{\kappa}}(X)} + (1 - \curvf{f}) \sum_{j \in X} f(j)$ and use this as
the surrogate function $\hat{f}$ for $f$. We choose $\hat{g}$ as $g$
itself. The surrogate problem becomes:
\begin{align}\label{EASSCeq}
\min \left\{\curvf{f} \sqrt{w^{f^{\kappa}}(X)} + (1 - \curvf{f}) \sum_{j \in X} f(j) \ \bigg\vert \ g(X) \geq c \right\}. 
\end{align}
While function $\hat{f}(X)=f^{\text{ea}}(X)$ is not modular, it is a
weighted sum of a concave over modular function and a modular
function. Fortunately, we can use the result
from~\cite{nikolova2010approximation}, where they show that any
function of the form of $\sqrt{w_1(X)} + w_2(X)$ can be optimized over
any polytope $\mathcal P$ with an approximation factor of $\beta (1 +
\epsilon)$ for any $\epsilon > 0$, where $\beta$ is the approximation
factor of optimizing a modular function over $\mathcal P$. The
complexity of this algorithm is polynomial in $n$ and
$\frac{1}{\epsilon}$. \arxiv{The main idea of their paper is to reduce
  the problem of minimizing $\sqrt{w_1(X)} + w_2(X)$, into $\log n$
  problems of minimizing a modular function over the polytope.} We use
their algorithm to minimize $f^{\text{ea}}(X)$ over the
\emph{submodular set cover} constraint and hence we call this algorithm EASSC. \arxiv{Again we have the two variants, \emph{primal
  EASSC} and \emph{dual EASSC}, which essentially use at every iteration the
primal and dual forms of set cover.}
\begin{theorem}\label{EAguar}
  \arxiv{The primal}
    EASSC obtains a guarantee of $O(\frac{\sqrt{n}\log n H_g}{ 1 +
    (\sqrt{n}\log n - 1)(1 - \curvf{f})} )$, where $H_g$ is the
  approximation guarantee of the set cover problem. \arxiv{Moreover, the dual EASSC obtains a bi-criterion approximation of
  $\left[O(\frac{\sqrt{n}\log n}{ 1 + (\sqrt{n}\log n - 1)(1 -
    \curvf{f})}), 1 - e^{-1}\right]$.}
\end{theorem}
\arxiv{
\begin{proof}
  The idea of the proof is to use the result
  from~\cite{nikolova2010approximation} where they show that any
  function of the form $\lambda_1 \sqrt{m_1(X)} + \lambda_2 m_2(X)$
  where $\lambda_1 \geq 0, \lambda_2 \geq 0$ and $m_1$ and $m_2$ are
  positive modular functions has a FPTAS, provided a modular function
  can easily be optimized over $\mathcal C$. Note that our function is
  exactly of that form. Hence, $\hat{f}(X)$ can be approximately
  optimized over $\mathcal C$. It now remains to show that this
  translates into the approximation guarantee. From
  Lemma~\ref{cor:learn}, we know that there exists a $\hat{f}$
  such that $\hat{f}(X) \leq f(X) \leq \beta(n) \hat{f}(X), \forall X$
  where
\begin{align}
\beta(n) = O\left(\frac{\sqrt{n} \log n}{(\sqrt{n} \log n - 1)(1 - \curvf{f}) + 1)}\right). 
\end{align}
Then, if $\hat{X}$ is the $1 + \epsilon$ approximately optimal
solution for minimizing $\hat{f}$ over $\{X: g(X) \geq c\}$, we
have that:
\begin{align}
  f(\hat{X}) 
  \leq \beta(n) \hat{f}(\hat{X}) 
  \leq H_g \beta(n) (1 +  \epsilon) \hat{f}(X^*) 
  \leq H_g \beta(n) (1 + \epsilon) f(X^*),
\end{align}
where $X^*$ is the optimal solution. We can set $\epsilon$ to any
constant, say $1$, and we get the result.
 
 The dual guarantee again follows in a very similar manner thanks to the guarantee for the dual SSC.
%The algorithm in~\cite{nikolova2010approximation} computes modular
%functions of the form $w_1(X) + \lambda w_2(X)$ for different values
%of $\lambda$ and the value of $H_g$ may involve taking the worst
%case factors over different values of $\lambda$. Many of the guarantees
%in Lemma~\ref{setcover}, however, depend only on $g$ and hence
%are independent of the value of $\lambda$.
\end{proof}
}

If the function $f$ has $\curvf{f} = 1$,
\JTR{This was originally $\curvf{f} = 1$ but I changed it to $\curvf{f} = 0$,
and this occurred a few other places in the paper. Note that according
to the def of $f^{\text{ea}}$, we get this simpler form $\eaxa{f}{X}$
only when $\curvf{f} = 0$, but we should double check this, and please
lets talk about it. The thing I don't understand is that
if $\curvf{f} = 0$ then $f$ is modular and we don't need the EA approximation.
Also see the comment tagged 5XXXXX below.}\RTJ{Fixed above as well as earlier.}
 we can use a much simpler
algorithm. In particular, \arxiv{since the ellipsoidal approximation is of
the form of $f^{\text{ea}}(X) = \eaxa{f}{X}$, }we can minimize $(f^{\text{ea}}(X))^2 =
w^f(X)$ at every iteration, giving a surrogate problem of the form
\arxiv{\begin{align}
\min\{ w^f(X) | g(X) \geq c\}.
\end{align}}
\notarxiv{$\min\{ w^f(X) | g(X) \geq c\}$.}
This is directly an instance of SSC, and in contrast to EASSC, we just need to solve SSC once. We call this algorithm EASSCc. \arxiv{This guarantee is tight up to $\log$ factors when $\curvf{f} = 1$.}
\begin{corollary}
\arxiv{The primal} EASSCc obtains an approximation guarantee of $O(\sqrt{n} \log n \sqrt{H_g})$. \arxiv{Similarly, the dual EASSCc obtains a bicriterion guarantee of $[O(\sqrt{n} \log n), 1 - e^{-1}]$.}
\end{corollary}
\arxiv{\begin{proof}
Let $\hat{X}$ be a set such that,
\begin{align}
w(\hat{X}) \geq H_g \min\{ w(X) | g(X) \geq c\} 
\end{align}
Then denote $\alpha(n) = O(\sqrt{n}\log n)$.
\begin{align}
f(\hat{X}) \leq \alpha(n) \sqrt{w(\hat{X})} \leq \sqrt{H_g} \alpha(n) \sqrt{w(X^*)} \leq \sqrt{H_g} \alpha(n) f(X^*)
\end{align}
In the above, $X^* = \argmin\{ f(X) | g(X) \geq c\}$. 

The result for the dual variant can also be similarly shown.
\end{proof}
}

\arxiv{
\begin{table}[tbh]
\small{
\begin{center}
\begin{tabular}{ |c|c|c|c| }
\hline
\multirow{2}{*}{Problem} & \multirow{2}{*}{Algorithm} & Approx. factor$^{*}$ & Hardness$^{*}$\\[1ex]
 & & Bi-Criterion factor$^{\#}$ & (Bi-Criterion $\frac{\sigma}{\rho})$$^{\#}$ \\[1ex] \hline 
\multirow{2}{*}{SSC} & Primal Greedy & $H(\max_j f(j))$$^{*}$ & $\log n$$^{*}$ \\[1ex]
 & Dual Greedy & $(1 + \epsilon, 1 - \frac{1}{e})$$^{\#}$ & $1 - 1/e$$^{\#}$\\[1ex]
\hline
\multirow{4}{*}{SCSC} & Primal ISSC & $\frac{n H_g}{1 + (n - 1)(1 - \curvf{f})}$$^{*}$ & \multirow{4}{*}{$\Omega(\frac{\sqrt{n}}{1 + (\sqrt{n} - 1)(1 - \curvf{f})})$$^{*, \#}$} \\ [1.5ex]
 & Dual ISSC & $[\frac{(1 + \epsilon)n}{1 + (n - 1)(1 - \curvf{f})}, 1 - 1/e]$$^{\#}$ & \\ [1.5ex]
& Primal EASSC & $O(\frac{\sqrt{n}\log n H_g}{1 + (\sqrt{n}\log n - 1)(1 - \curvf{f})})$$^{*}$ & \\ [1.5ex]
& Dual EASSC & $[O(\frac{\sqrt{n}\log n}{1 + (\sqrt{n}\log n - 1)(1 - \curvf{f})}, 1 - 1/e]$$^{\#}$ & \\ [1.5ex]
& Primal EASSCc & $O(\sqrt{n}\log n \sqrt{H_g})$$^{*}$ & \\ [1ex]
& Dual EASSCc & $[O(\sqrt{n}\log n), 1 - 1/e]$$^{\#}$ & \\ [1ex]
\hline
SK & Greedy & $1 - 1/e$$^{*}$ & $1 - 1/e$$^{*}$ \\ [1ex]
\hline
\multirow{3}{*}{SCSK} & Greedy & $1/n$$^{*}$ & \multirow{4}{*}{$\Omega(\frac{\sqrt{n}}{1 + (\sqrt{n} - 1)(1 - \curvf{f})})$$^{*, \#}$} \\ [1.5ex]
& ISK & $[1 - e^{-1}, \frac{K_f}{1 + (K_f - 1)(1 - \curvf{f})}]$$^{\#}$ & \\ [1.5ex]
& Primal EASK & $[1 + \epsilon, O(\frac{\sqrt{n}\log n H_g}{1 + (\sqrt{n}\log n - 1)(1 - \curvf{f})})]$$^{\#}$ & \\ [1.5ex]
& Dual EASK & $[(1 - 1/e)(1 + \epsilon), O(\frac{\sqrt{n}\log n}{1 + (\sqrt{n}\log n - 1)(1 - \curvf{f})})]$$^{\#}$ & \\ [1.5ex]
& EASKc & $[1 - 1/e, O(\sqrt{n}\log n)]$$^{\#}$ & \\ [1ex]
\hline
\end{tabular}
\end{center}
\caption{Worst case approximation factors, hardness for Problems 1 
  and 2 and their special cases.} }
 \end{table}}
\normalsize

%%%%%%%%%%%%%%%%%%%%%%%%%%%%%%%%%%%%%%%%%%%%%%%%%%%%%%%%%%%%%%%%
\subsection{Approximation Algorithms for SCSK}
\label{sec:appr-algor-scsk}

In this section, we describe our approximation algorithms for SCSK. We
note the dual nature of the algorithms in this current section to those
given in \S\ref{sec:appr-algor-scsc}. We first investigate a special case, the submodular knapsack (SK), and then provide three algorithms, two of them (Gr and ISK) being practical with slightly weaker theoretical guarantee, and another one (EASK) which is not scalable but has the tightest guarantee.\looseness-1

\textbf{Submodular Cost Knapsack (SK): }
We start with a special case of SCSK (Problem 2), where $f$ is a modular
function and $g$ is a submodular function. In this case, SCSK turns
into the SK problem for which the greedy
algorithm with partial enumeration provides a $1 - e^{-1}$
approximation~\cite{sviridenko2004note}. The greedy algorithm can be
seen as an instance of Algorithm~\ref{alg:framework} with $\hat{g}$
being the modular lower bound of $g$ and $\hat{f}$ being $f$, which is
already modular. \arxiv{In particular, we then get back the framework
  of~\cite{rkiyersemiframework2013}, where the authors show that
  choosing a permutation based on a greedy ordering, exactly analogous
  to Eqn.~\eqref{greedyssc}, provides the bounds.} In particular,
define:
 \begin{align}
\label{greedyknapsack}
 \perm(i) \in 
 \argmax \left\{\frac{g(j | S^{\perm}_{i-1})}{f(j)} \ \bigg\vert \ {j \notin S^{\perm}_{i-1}, 
          f(S^{\perm}_{i-1} \cup \{j\}) \leq b}\right\},
 \end{align} 
 where the remaining elements are chosen arbitrarily.  \arxiv{A slight
   catch however is that for the analysis to work, \cite{sviridenko2004note} needs to
   consider $n \choose 3$ instances of such orderings (partial
   enumeration), chosen by fixing the first three elements in the
   permutation~\cite{rkiyersemiframework2013}. We can however just choose
   the simple greedy ordering in one stage, to get a slightly worse
   approximation factor of $1 -
   e^{-1/2}$~\cite{rkiyersemiframework2013, linbudget,krause05note}.
}  \arxivalt{
\begin{lemma}\label{greedybudget}~\cite{rkiyersemiframework2013}
Choosing the surrogate function $\hat{f}$ as $f$ and $\hat{g}$ as $h^{\perm}$ in Algorithm~\ref{alg:framework} yields a set $X^g$: 
\begin{equation}
  \label{eq:8}
  \max\{\max_{i: f(i) \leq b} g(i), g(X^g)\} \geq 1/2(1 - 1/e)g(X^*).
\end{equation}
Let $\perm^{ijk}$ be a permutation with $i,j,k$ in the first three
positions, and the remaining arrangement greedy. Running $O(n^3)$
restarts of one iteration of Algorithm~\ref{alg:framework} yields sets
$X_{ijk}$ with
\begin{equation}
  \label{eq:9}
  \max_{i,j,k \in V} f(X_{ijk}) \,\geq (1-1/e) f(X^*).
\end{equation} 
\end{lemma}}
{% A simpler version of this Lemma in the main paper.
The following is an informal description of the result described formally in~\cite{nipsextendedvsubcons}.
\begin{lemma}
Choosing the surrogate function $\hat{f}$ as $f$ and $\hat{g}$ as $h^{\perm}$ (with $\perm$ defined in eqn~\eqref{greedyknapsack}) in Algorithm~\ref{alg:framework} with appropriate initialization obtains a guarantee of $1 - 1/e$ for SK.
\end{lemma}
}

\textbf{Greedy (Gr): } 
A similar greedy algorithm can provide approximation guarantees for
the general SCSK problem, with submodular $f$ and $g$. Unlike the
knapsack case in~\eqref{greedyknapsack}, however, at iteration $i$ 
we choose an element $j \notin S_{i-1} :
f(S^{\perm}_{i-1} \cup \{j\}) \leq b$ which maximizes $g(j |
S_{i-1})$. In terms of Algorithm~\ref{alg:framework}, this is
analogous to choosing a permutation, $\perm$ such that:
\begin{align}
\perm(i) \in \argmax \{g(j | S^{\perm}_{i-1}) | {j \notin S^{\perm}_{i-1}, f(S^{\perm}_{i-1} \cup \{j\}) \leq b}\}. 
\end{align}
\begin{theorem}\label{GrSCSK}
The greedy algorithm for SCSK obtains an approx. factor of $\frac{1}{\kappa_g} (1 - (\frac{K_f - \kappa_g}{K_f})^{k_f}) \geq \frac{1}{K_f}$, where $K_f = \max\{|X|: f(X) \leq b\}$ and $k_f = \min\{|X|: f(X) \leq b \ \& \ \forall j \in X, f(X \cup j) > b\}$.
\end{theorem}
\arxiv{
\begin{proof}
The proof of this result follows directly from~\cite{conforti1984submodular, rkiyersemiframework2013}. In particular, it holds for any down monotone constraint. It is easy to see that the constraint $\{f(X) \leq b\}$ is down-monotone when $f$ is a monotone submodular function.
\end{proof}
} 
In the worst case, $k_f = 1$ and $K_f = n$, in which case the
guarantee is $1/n$. The bound above follows from a simple observation
that the constraint $\{f(X) \leq b\}$ is down-monotone for a monotone
function $f$. However, in this variant, we do not use any specific
information about $f$. In particular it holds for maximizing a
submodular function $g$ over any down monotone
constraint~\cite{conforti1984submodular}. Hence it is
conceivable that an algorithm that uses both $f$ and $g$ to choose the
next element could provide better bounds. We do not, however, currently have
the analysis for this.\looseness-1

\textbf{Iterated Submodular Cost Knapsack (ISK): } Here, we
choose $\hat{f_t}(X)$ as a modular upper bound of $f$, tight at $X^t$.
Let $\hat{g_t} = g$. Then at every iteration, we solve\arxivalt{:
\begin{align}
\max\{ g(X) | m^f_{X^t}(X) \leq b\},
\end{align}}{ $\max\{ g(X) | m^f_{X^t}(X) \leq b\}$,}
which is a submodular maximization problem subject to a knapsack
constraint (SK). As mentioned above, greedy can solve this nearly optimally.  We start with $X^0 = \emptyset$, choose $\hat{f_0}(X) =
\sum_{j \in X} f(j)$ and then iteratively continue this process until
convergence (note that this is an ascent algorithm). We
have the following theoretical guarantee:
\begin{theorem}\label{ISKapprox}
  Algorithm ISK obtains a set $X^t$ such that $g(X^t) \geq (1 - e^{-1})
  g(\tilde{X})$, where $\tilde{X}$ is the optimal solution of $\max\left\{g(X) \ | \ f(X) \leq \frac{b (1 + (K_f - 1)(1 -
    \curvf{f})}{K_f}\right\}$ and where $K_f = \max\{|X|: f(X) \leq b\}$.\looseness-1
\end{theorem}
\arxiv{\begin{proof}
We are given that $\tilde{X}$ is the optimal solution to the problem:
\begin{align}
\max\left\{g(X) \ | \ f(X) \leq \frac{b (1 + (K_f - 1)(1 - \curvf{f})}{K_f}\right\}
\end{align}
$\tilde{X}$ is also a feasible solution to the problem:
\begin{align}\label{subprobKI}
\max\left\{g(X) \ | \ \sum_{j \in X} f(j) \leq b \right\}
\end{align}
The reason for this is that:
\begin{align}
\sum_{j \in X} f(j) \leq \frac{K_f}{1 + (K_f - 1)(1 - \curvf{f})} f(X) \leq  \frac{K_f}{1 + (K_f - 1)(1 - \curvf{f})}\frac{b (1 + (K_f - 1)(1 - \curvf{f}))}{K_f} \leq b \nonumber
\end{align}
\JTR{fix bugs in the above equation, both the missing parentheses and also
  the divide/multiply by $b$.}\RTJ{done}  Now, at the first iteration we are
guaranteed to find a set $X^1$ such that $g(X^1) \geq (1 -
1/e)g(\tilde{X})$. The further iterations will only improve the
objective since this is a ascent algorithm.
\end{proof}}
It is worth pointing out that the above bound holds even after the
first iteration of the algorithm. It is interesting to note the similarity
between this approach and \arxiv{the iterated submodular set cover}
ISSC. \notarxiv{Notice that the guarantee above is not a standard bi-criterion approximation. We show in the extended version~\cite{nipsextendedvsubcons} that with a simple transformation, we can obtain a bicriterion guarantee.}\arxiv{Notice that the bound above is a form of a bicriterion approximation factor of type-II. In particular, we obtain a feasible solution at the end of the algorithm. We can obtain a type-I bicriterion approximation bound, by running this for a larger budget constraint. In particular, we run the above algorithm with a budget constraint of  $\frac{b K_f}{1 + (K_f - 1)(1 - \curvf{f})}$ instead of $b$. The following guarantee then follows.

\begin{lemma}
The ISK algorithm of type-I is guaranteed to obtain a set $X$ which has a bicriterion approximation factor of $[1 - e^{-1}, \frac{K_f}{1 + (K_f - 1)(1 - \curvf{f})}]$, where $K_f = \max\{|X|: f(X) \leq b\}$.
\end{lemma}
\begin{proof}
This result directly follows from the Lemma above, and the transformation from a type-II approximation algorithm to a type-I one.
\end{proof}
}

\textbf{Ellipsoidal Approximation based Submodular Cost Knapsack (EASK): }
Choosing the Ellipsoidal Approximation $f^{ea}$ of $f$ as a surrogate function, we obtain a simpler problem:
\begin{align}\label{EASKeq}
\max\left\{ g(X) \ \bigg\vert \ \curvf{f} \sqrt{w^{f^{\kappa}}(X)} + (1 - \curvf{f}) \sum_{j \in X} f(j) \leq b\right\}. 
\end{align}
In order to solve this problem, we look at its dual
problem (i.e., Eqn.~\eqref{EASSCeq}) and use Algorithm~\ref{alg:alg1} to convert
the guarantees.  \arxivalt{Recall that
Eqn.~\eqref{EASSCeq} itself admits two variants of the ellipsoidal approximation,
which we had referred to as the primal EASSC and the dual EASSC. Hence, we call the combined algorithms, primal and dual EASK, which admit the following guarantees:}{%not arxiv part
We call this procedure EASK. We then obtain guarantees very similar to Theorem~\ref{EAguar}.}
\begin{lemma}
  \arxiv{The primal} EASK obtains a guarantee of $\left[1 +
  \epsilon, O(\frac{\sqrt{n}\log n H_g}{ 1 + (\sqrt{n}\log n - 1)(1 -
    \curvf{f})})\right]$. \arxiv{Similarly, the dual EASK obtains a
  guarantee of $\left[(1 + \epsilon)(1 - 1/e), O(\frac{\sqrt{n}\log n}{ 1 +
    (\sqrt{n}\log n - 1)(1 - \curvf{f})})\right]$.}
\end{lemma}
\arxiv{\begin{proof}This Lemma directly follows from
  Theorem~\ref{thm1} and
  Theorem~\ref{EAguar}. \end{proof}
These factors are akin to
those of Theorem~\ref{EAguar}, except for the additional factor of $1 +
\epsilon$. Unlike the greedy algorithm, ISSC and ISK, note that both EASK and EASSC are
enormously costly and complicated algorithms. Hence, it also seems at first
thought that EASK would need to run multiple versions of EASSC at each conversion round of Algorithm~\ref{alg:alg1}. Fortunately, however, we need to compute the
Ellipsoidal Approximation just once and the algorithm can reuse it for
different values of $c$. Furthermore, since construction
of the approximation is often the bottleneck, this scheme is likely to be as costly
as EASSC in practice.\looseness-1
}

In the case when the submodular function has a curvature $\curvf{f} =
1$, we can actually provide a simpler algorithm without needing to
use the conversion algorithm (Algorithm~\ref{alg:alg1}). In this case, we
can directly choose the ellipsoidal approximation of $f$ as
$\sqrt{w^f(X)}$ and solve the surrogate problem:\arxivalt{
\begin{align} 
\label{EAKeqn}
\max\{g(X) : w^f(X) \leq b^2\}. 
\end{align}}{ $\max\{g(X) : w^f(X) \leq b^2\}.$}
This surrogate problem is a submodular cost knapsack problem, which we can
solve using the greedy algorithm. We call this algorithm EASKc. This guarantee is tight up to $\log$ factors if $\curvf{f} = 1$.
\begin{corollary}
Algorithm EASKc obtains a bi-criterion guarantee of $[1 - e^{-1}, O(\sqrt{n}\log n)]$.
\end{corollary}
\arxiv{\begin{proof} 
Let $\hat{X}$ be the approximate optimizer of
    Eqn.~\eqref{EAKeqn}.  First we show that $g(\hat{X}) \geq (1 -
    1/e) g(X^*)$ where $X^*$ is the optimizer of problem 2. Notice
    that $f(X^*) \leq b$ since it feasible in problem 2. Also,
    $\sqrt{w^f(X^*)} \leq f(X^*) \leq b$ and hence $X^*$ is feasible
    in Eqn.~\eqref{EAKeqn}. Hence it holds that $g(\hat{X}) \geq (1 -
    1/e) g(X^*)$.

We now show that $f(\hat{X}) \leq b\sqrt{n}\log n$. Notice that,
\begin{align}
f(\hat{X}) \leq \sqrt{n} \log n \sqrt{w^f(\hat{X})} \leq b \sqrt{n} \log n
\end{align}
\end{proof}}

\subsection{Extensions beyond SCSC and SCSK}\label{extentions}
\arxivalt{SCSC is in fact more general and can be extended
  to more flexible and complicated constraints which can arise naturally in many
  applications~\cite{krause08robust, guillory2011simultaneous}. Notice
  first that
\begin{align}
\{g(X) \geq \alpha\} \Leftrightarrow \{g^{\prime}(X) = g^{\prime}(V)\}
\end{align}
 where $g^{\prime}(X) = \min\{g(X), \alpha\}$. We can also have ``and'' constraints as $g_1(X) = g_1(V)$ and $g_2(X) = g_2(V)$. These have a simple equivalence:
\begin{align}
\{g_1(X) = g_1(V) \wedge g_2(X) = g_2(V)\} \Leftrightarrow 	\{g(X) = g(V)\}
\end{align}
 when $g(X) = g_1(X) + g_2(X)$~\cite{krause08robust}. Moreover, we can also handle $k$ `and' constraints, by defining $g(X) = \sum_{i = 1}^k g_i(X)$.

Similarly we can have ``or'' constraints, i.e., $g_1(X) = g_1(V)$ or $g_2(X) = g_2(V)$. These also
have a nice equivalence: 
\begin{align}
\{g_1(X) = g_1(V) \vee g_2(X) = g_2(V)\} \Leftrightarrow \{g(X) = g(V)\} 
\end{align}
by defining $g(X) = g_1(X)g_2(V) + g_2(X)g_1(V) -
g_1(X)g_2(X)$~\cite{guillory2011simultaneous}. We can also extend these recursively to multiple `or' constraints. Hence our algorithms
can directly solve all these variants of SCSC.
%\JTR{What about ``not'' style constraints, 
%of the form $g_1(X) \geq c_1 \wedge g_2(X) \leq c_2$? It
%seems this should be possible to handle by somehow incorporating $g_2$ into $f$.
%For example, consider the function $g'_2(X) = Lg_2(V) + L \min( g_2(X) - c_2, 0)$
%when $L$ is very big, 
%then $g'_2(X)$ is monotone submodular, 
%and then rather than $f$, we minimize
%$f'(X) = f(X) + g'_2(X)$ s.t.\ $g_1(X) \leq c_1$. If we have
%approximate quality in terms of $f'(X)$ ($f'(X) \leq \sigma f(X^*)$) and if $L$ is big enough,
%I think it might be possible to show that we must in such case have $g_2(X) \leq c_2$.
%The problem however is that this will probably swamp any approximation we had
%in terms of the original $f$, but maybe there is some way of doing it. Having
%``not'' style constraints and some form of approximate quality would be very powerful.}

SCSK can also be extended to handle more complicated forms
of functions $g$. In particular, consider the function
\begin{align}
g(X)= \min\{g_1(X), g_2(X), \dots, g_k(X)\}
\end{align} 
where the functions $g_1, g_2,
\dots, g_k$ are submodular. Although $g(X)$ in this case is not
submodular, this scenario occurs naturally in many applications,
particularly sensor placement~\cite{krause08robust}. The problem in~\cite{krause08robust} is in fact a special case of Problem 2, using a modular function $f$. Often, however,
the budget functions involve a cooperative cost, in which case $f$ is submodular.  Using Algorithm~\ref{alg:alg1},
however, we can easily solve this by iteratively solving the dual
problem. Notice that the dual problem is in the form of Problem 1 with a non-submodular
constraint $g(X) \geq c$. It is easy to see that this is equivalent to
the constraint $g_i(X) \geq c, \forall i = 1, 2, \dots, k$, which can be solved thanks to the techniques above. \looseness-1

We can also handle multiple constraints in SCSK. In particular, consider multiple `and' constraints -- $\{f_i(X) \leq b_i, i = 1, 2, \cdots, k\}$, for monotone submodular functions $f_i$ and $g$. A first observation is that the greedy algorithm can almost directly extended to these cases, since we do not use any specific property of the constraints, while providing the approximation guarantees. Hence Theorem~\ref{GrSCSK} can directly be extended to this case. We can also provide bi-criterion approximation guarantees with `and' constraints. The approximate feasibility for a $[\rho, \sigma]$ bi-criterion approximation in this setting would be to have a set $X$ such that $f_i(X) \leq \sigma b_i$. Algorithm ISK can easily then be used in this scenario, and at every iteration we would solve a monotone submodular maximization problem subject to multiple linear (or knapsack constraints). Surprisingly this problem also has a constant factor ($1 - 1/e$) approximation guarantee~\cite{kulik2009maximizing}. Hence we can retain the same approximation guarantees as in Theorem~\ref{ISKapprox}. We can also use algorithm EASKc, and obtain a curvature-independent approximation bound for this problem. The reason for this is the Ellipsoidal Approximation is of the form $\sqrt{w^{f_i}(X)}$, for each $i$ and squaring it will lead to knapsack constraints. If we add the curvature terms (i.e., try to implement EASK in this setting), we obtain a much more complicated class of constraints, which we do not currently know how to handle.

We can also extend SCSC and SCSK to non-monotone submodular functions. In particular, recall that the submodular knapsack has constant factor approximation guarantees even when $g$ is non-monotone submodular~\cite{fiege2011submodmax}. Hence, we can obtain bi-criterion approximation guarantees for the Submodular Set Cover (SSC) problem, by solving the Submodular Knapsack (SK) problem multiple times. We can similarly do SCSC and SCSK when $f$ is monotone submodular and $g$ is non-monotone submodular, by extending ISSC, EASSC, ISK and EASK (note that in all these cases, we need to solve a submodular set cover or submodular knapsack problem with a non-monotone $g$). These algorithms, however, do not extend if $f$ is non-monotone, and we do not currently know how to implement these.
}{SCSC and SCSK can in fact be extended
  to more flexible and complicated constraints which can arise naturally in many
  applications~\cite{krause08robust, guillory2011simultaneous}. These include multiple covering and knapsack constraints -- i.e., $\min\{f(X) | g_i(X) \geq c_i, i = 1, 2, \cdots k\}$ and $\max\{g(X) | f_i(X) \leq b_i, i = 1, 2, \cdots k\}$, and robust optimization problems like $\max\{\min_i g_i(X) | f(X) \leq b\}$, where the functions $f, g, f_i$'s and $g_i$'s are submodular. We also consider SCSC and SCSK with non-monotone submodular functions. Due to lack of space, we defer these discussions to the extended version of this paper~\cite{nipsextendedvsubcons}.}
%%%%%%%%%%%%%%%%%%%%%%%%%%%%%%%%%%%%%%%%%%%%%%%%%%%%%%%%%%%%%%%%%%%
\subsection{Hardness}
\label{sec:hardness}

In this section, we provide the hardness for Problems 1 and 2. The
lower bounds serve to show that the approximation factors above are
almost tight.
\begin{theorem}
 For any $\kappa > 0$, there exists submodular functions with curvature $\kappa$ such that no polynomial time algorithm for Problems 1 and 2
  achieves a bi-criterion factor better than
  $\frac{\sigma}{\rho} = \frac{n^{1/2 - \epsilon}}{1 + (n^{1/2 -
      \epsilon} - 1)(1 - \kappa)}$ for any $\epsilon > 0$. \looseness-1
\end{theorem}
\arxiv{
\begin{proof}
  We prove this result using the hardness construction
  from~\cite{goemans2009approximating, svitkina2008submodular}. The
  main idea of their proof technique is to construct two submodular
  functions $f(X)$ and $f_R(X)$ that with high probability are
  indistinguishable. Thus, also with high probability, no algorithm
  can distinguish between the two functions and the gap in their
  values provides a lower bound on the approximation. We shall see
  that this lower bound in fact matches the approximation factors up
  to log factors and hence this is the hardness of the problem.

  Define two monotone submodular functions $f(X) = \curvf{f}\min\{|X|,
  \alpha\} + (1 - \curvf{f})|X|$ and $f_R(X) = \curvf{f}\min\{\beta +
  |X \cap \bar{R}|, |X|, \alpha\} + (1 - \curvf{f}) |X|$, where $R
  \subseteq V$ is a random set of cardinality $\alpha$. Let $\alpha$
  and $\beta$ be an integer such that $\alpha = x\sqrt{n}/5$ and
  $\beta = x^2/5$ for an $x^2 = \omega(\log n)$. Both $f$ and $f_R$
  have curvature $\curvf{f}$. We also assume a very simple
  function $g(X) = |X|$. Given an arbitrary $\epsilon>0$, set $x^2 =
  n^{2\epsilon} = \omega(\log n)$. Then the ratio between
  $f_R^{\curvf{f}}(R)$ and $g^{\curvf{f}}(R)$ is
  $\frac{n^{1/2-\epsilon}}{1 + (n^{1/2-\epsilon}-1)(1 -
    \curvf{f})}$. A Chernoff bound analysis very similar
  to~\cite{svitkina2008submodular} reveals that any algorithm that
  uses a polynomial number of queries can distinguish $h^{\kappa}$ and
  $f^{\kappa}_R$ with probability only $n^{-\omega(1)}$, and therefore
  cannot reliably distinguish the functions with a polynomial number
  of queries.

  We first prove the first part of this theorem (i.e., for SCSC). In
  this case, we consider the problem
\begin{align}\label{probststproof}
\min\{h(X) | |X| \geq \alpha\}
\end{align}
with $h$ chosen as $f$ and $f_R$ respectively. It is easy to see that
$R$ is the optimal set in both cases. However the ratio between the
two is
\begin{align}
\gamma(n) = \frac{n^{1/2-\epsilon}}{1 + (n^{1/2-\epsilon}-1)(1 - \curvf{f})}
\end{align}
Now suppose there exists an algorithm which is guaranteed to obtain an
approximation factor better than $\gamma(n)$. Then by running this
algorithm, we are guaranteed to find different answers for
Eqn~\eqref{probststproof} above. This must imply that this algorithm
can distinguish between $f_R$ and $g$ which is a contradiction. Hence
no algorithm for Problem 1 can obtain an approximation guarantee
$\frac{n^{1/2 - \epsilon}}{1 + (n^{1/2 - \epsilon} - 1)(1 -
  \curvf{f})}$.

In order to extend the result to the bi-criterion case, we need to
show that no bi-criterion approximation algorithm can obtain a factor
better than $\frac{\sigma}{\rho} = \gamma(n)$. Assume there exists a
bi-criterion algorithm with a factor $[\sigma, \rho]$ such that:
\begin{align}
\frac{\sigma}{\rho} < \gamma(n) = \frac{n^{1/2-\epsilon}}{1 + (n^{1/2-\epsilon}-1)(1 - \curvf{f})}
\end{align}
Then, we are guaranteed to obtain a set $S$ in
equation~\eqref{probststproof} such that $h(S) \leq \sigma OPT$ and
$|S| \geq \rho \alpha$, where $\sigma \geq 1$ and $\rho \leq 1$. Now
we run this algorithm with $h = f$ and $h = f_R$ respectively. With $h
= f$, it is easy to see that the algorithm obtains a set $S_1$ such
that $f(S_1) \leq \sigma \alpha$ and $|S_1| \geq \rho
\alpha$. Similarly, with $h = f_R$, the algorithm finds a set $S_2$
such that $f(S_2) \leq \sigma (\kappa \beta + (1 - \kappa)\alpha)$ and
$|S_2| \geq \rho \alpha$. Since $f_R$ and $f$ are indistinguishable,
$S_1 = S_2 = S$.

We first assume that $|S| < \alpha$. Then $f(S) = |S| \geq \rho \alpha$. We then have that,
\begin{align}
f_R(S) &\leq \sigma (\kappa \beta + (1 - \kappa)\alpha \\
	&\leq \sigma n^{2\epsilon} (1 + (n^{1/2 - \epsilon} - 1)(1 - \curvf{f}) \\
	&< \rho \gamma(n) n^{2\epsilon} (1 + (n^{1/2 - \epsilon} - 1)(1 - \curvf{f}) \\
	&< \rho n^{1/2 + \epsilon} \\
	&< \rho \alpha
\end{align}
This implies that the algorithm can distinguish between $f_R$ and $f$
which is a contradiction. Now consider the second case when $|S| \geq
\alpha$. In this case, $f(S) = \alpha$. Again the chain of
inequalities above shows that $f_R(S) < \rho \alpha < \alpha$ since
$\rho \leq 1$. This again implies that the algorithm can distinguish
between $f_R$ and $f$ which is a contradiction. Hence no bi-criterion
approximation algorithm can obtain a factor better than $\gamma(n)$.

To show the second part (for SCSK), we can simply invoke Theorem~\ref{thm1} and argue that any $[\rho, \sigma]$ bi-criterion approximation algorithm for problem 2 with 
\begin{align}
\frac{\sigma}{\rho} < \frac{\gamma(n)}{1 + \epsilon}
\end{align}
 can be used in algorithm 3, to provide a $[\sigma, \rho]$ bi-criterion algorithm with $\frac{\sigma}{\rho} < \gamma(n)$. This contradicts the first part above and hence the same hardness applies to problem 2 as well.
\end{proof}
} The above result shows that EASSC and EASK meet the bounds above to
log factors. We see an interesting curvature-dependent influence
on the hardness.  \arxivalt{We also see from our approximation guarantees
  also that the curvature of $f$ plays a more influential role than
  the curvature of $g$ on the approximation quality.}{We also see this
  phenomenon in the approximation guarantees of our algorithms.} In
particular, as soon as $f$ becomes modular, the problem becomes easy,
even when $g$ is submodular. This is not surprising since the
submodular set cover problem and the submodular cost knapsack problem
both have constant factor guarantees.\looseness-1

\arxiv{\begin{figure}[t]
  \centering
\hspace{-10pt}
\includegraphics[width=0.45\textwidth]{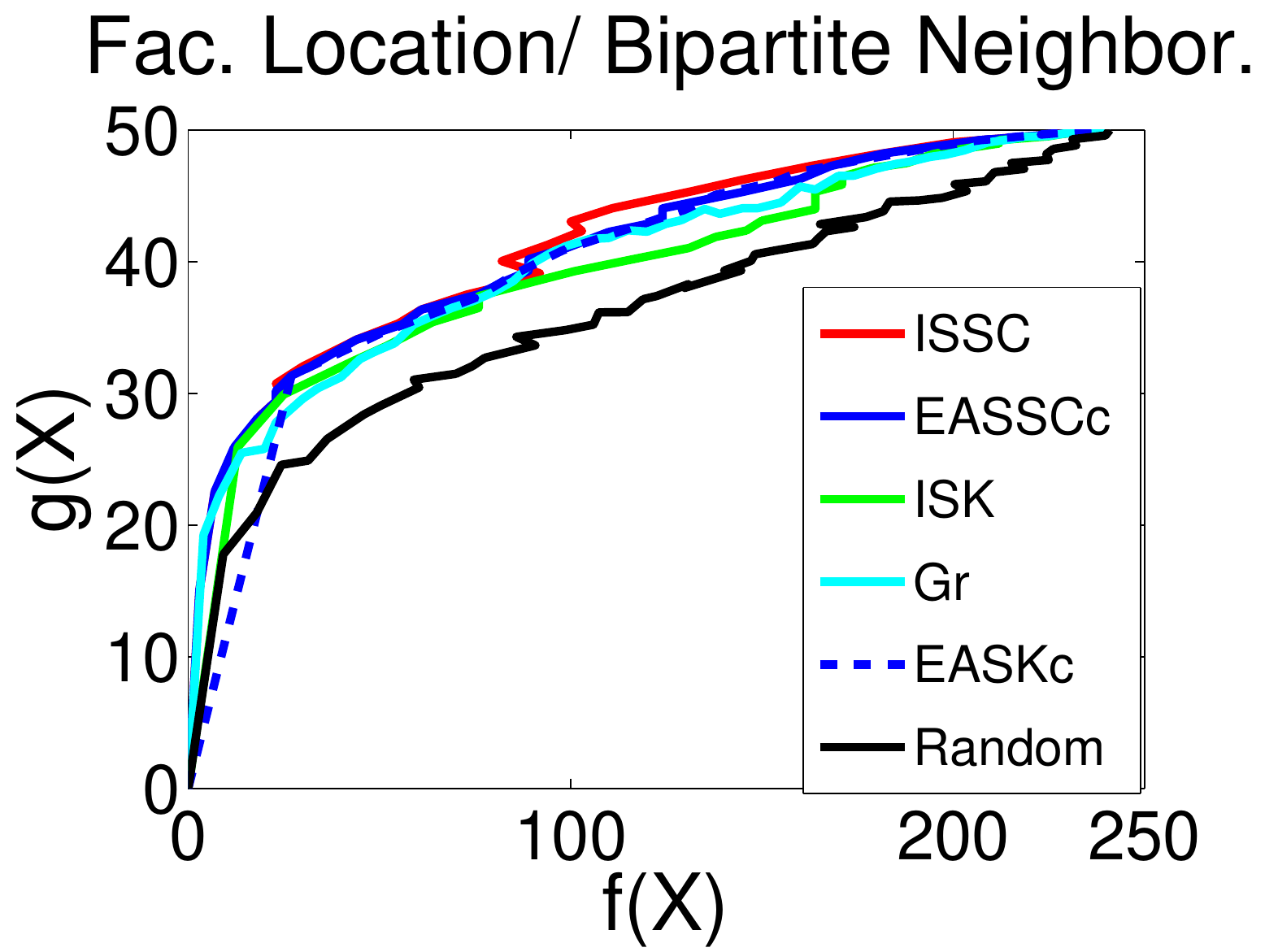}\hspace{-10pt} 
  ~ 
\includegraphics[width=0.45\textwidth]{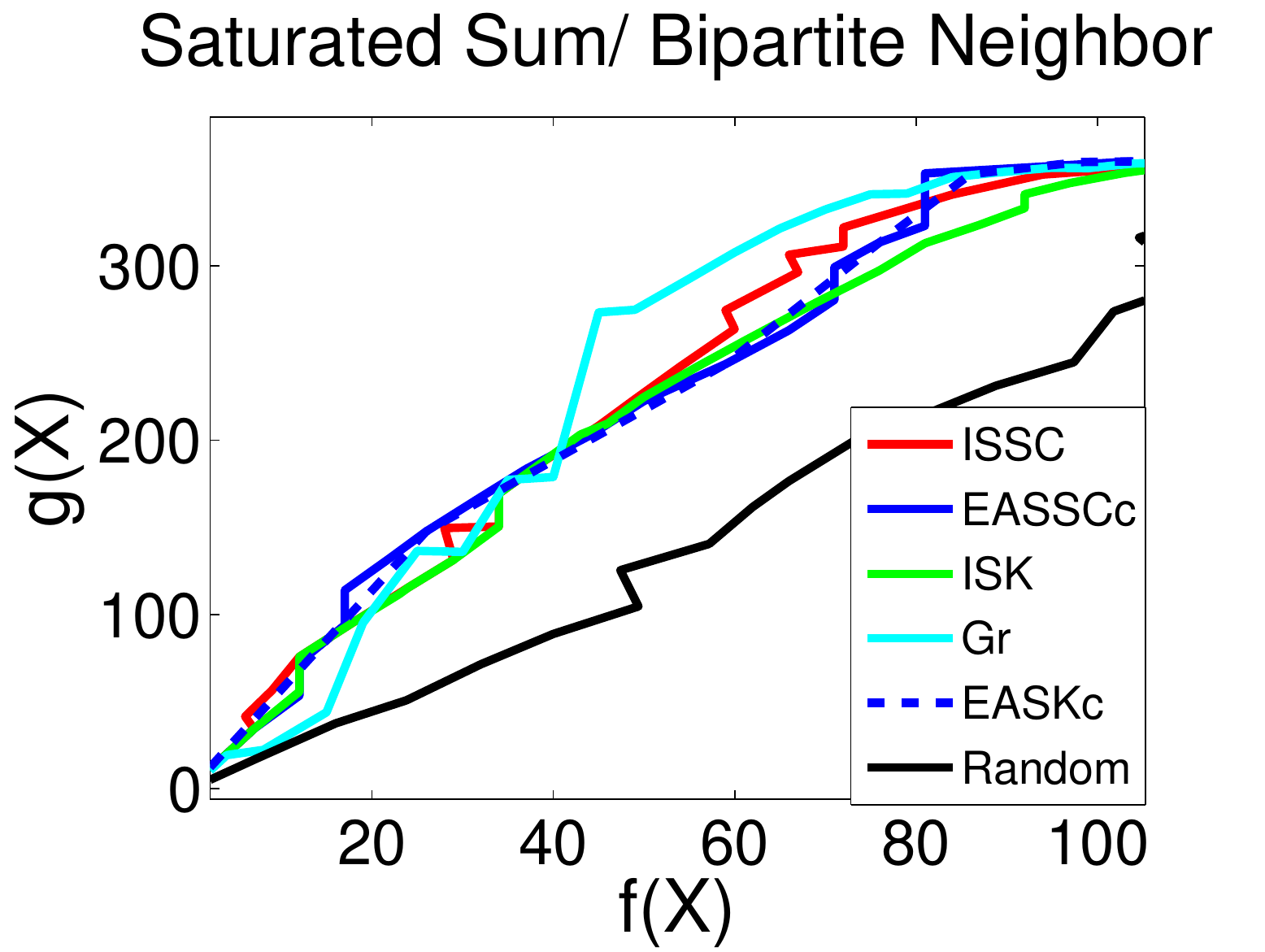} \hspace{-10pt} 
\caption{The two figures show the performance of the algorithms in the text on
  Problems 1 and 2 (in both cases, higher the better).
  \JTR{change figure legend, I changed MI to ISSC. Also EAk should be EAK,
    and KI needs to be changed as well.}
}
  \label{fig:sim}
\end{figure}}

%%%%%%%%%%%%%%%%%%%%%%%%%%%%%%%%%%%%%%%%%%%%%%%%%%%%%%%%%%%%%%%%%
\section{Experiments}
\label{sec:experiments}

In this section, we empirically compare the performance of the various
algorithms discussed in this paper. We are motivated by the speech
data subset selection application\arxivalt{~\cite{lin2009select,lin11,jegelkanips}}{~\cite{lin2009select,lin11}} with the submodular function $f$ encouraging limited
vocabulary while $g$ tries to achieve acoustic variability. A natural
choice of the function $f$ is a function of the form
$|\Gamma(X)|$, where $\Gamma(X)$ is the neighborhood function
on a bipartite graph constructed between the utterances and the words~\cite{lin11}. For the coverage function
$g$, we use two types of coverage: one is a facility
location function \arxivalt{
\begin{align}
g_1(X) = \sum_{i \in V} \max_{j \in X} s_{ij}
\end{align}
}{$g_1(X) = \sum_{i \in V} \max_{j \in X} s_{ij}$} while the other is a saturated sum function
\arxivalt{
\begin{align}
g_2(X) = \sum_{i \in V}\min\{\sum_{j \in X} s_{ij}, \alpha \sum_{j \in V} s_{ij}\}.
\end{align}
}{$g_2(X) = \sum_{i \in V}\min\{\sum_{j \in X} s_{ij}, \alpha \sum_{j \in V} s_{ij}\}$.} Both these functions are defined in terms of a similarity matrix
$\mathbf S = \{s_{ij}\}_{i, j \in V}$, which we define on the TIMIT corpus~\cite{timit}, using the string kernel metric~\cite{rousu2006efficient} for similarity. Since some of our algorithms, like the Ellipsoidal Approximations, are computationally intensive, we restrict ourselves to $50$ utterances.

%We first evaluate these
%coverage functions in the context of SSC (i.e., $\min\{w(X) | g(X)
%\geq c\}$) for a random modular function $w$ and the similarity matrix from the TIMIT corpus. We compare the primal
%and dual variants of the greedy algorithm and sweep through the
%different values of $c$. The results are shown in
%Figures~\ref{fig:sim}(a, b). We see that though the dual variant has
%better bicriterion guarantees, it performs slightly poorly empirically
%compared to the primal variant. Moreover since it runs the knapsack
%greedy algorithm multiple times, it is much slower in
%practice.\looseness-1

\narxiv{
\captionsetup[figure]{font=small,skip=0pt}
\begin{wrapfigure}[8]{r}{0.5\textwidth}
\vspace{-2.5ex}
\begin{minipage}{0.5\textwidth}
\centering
\includegraphics[width=0.45\textwidth]{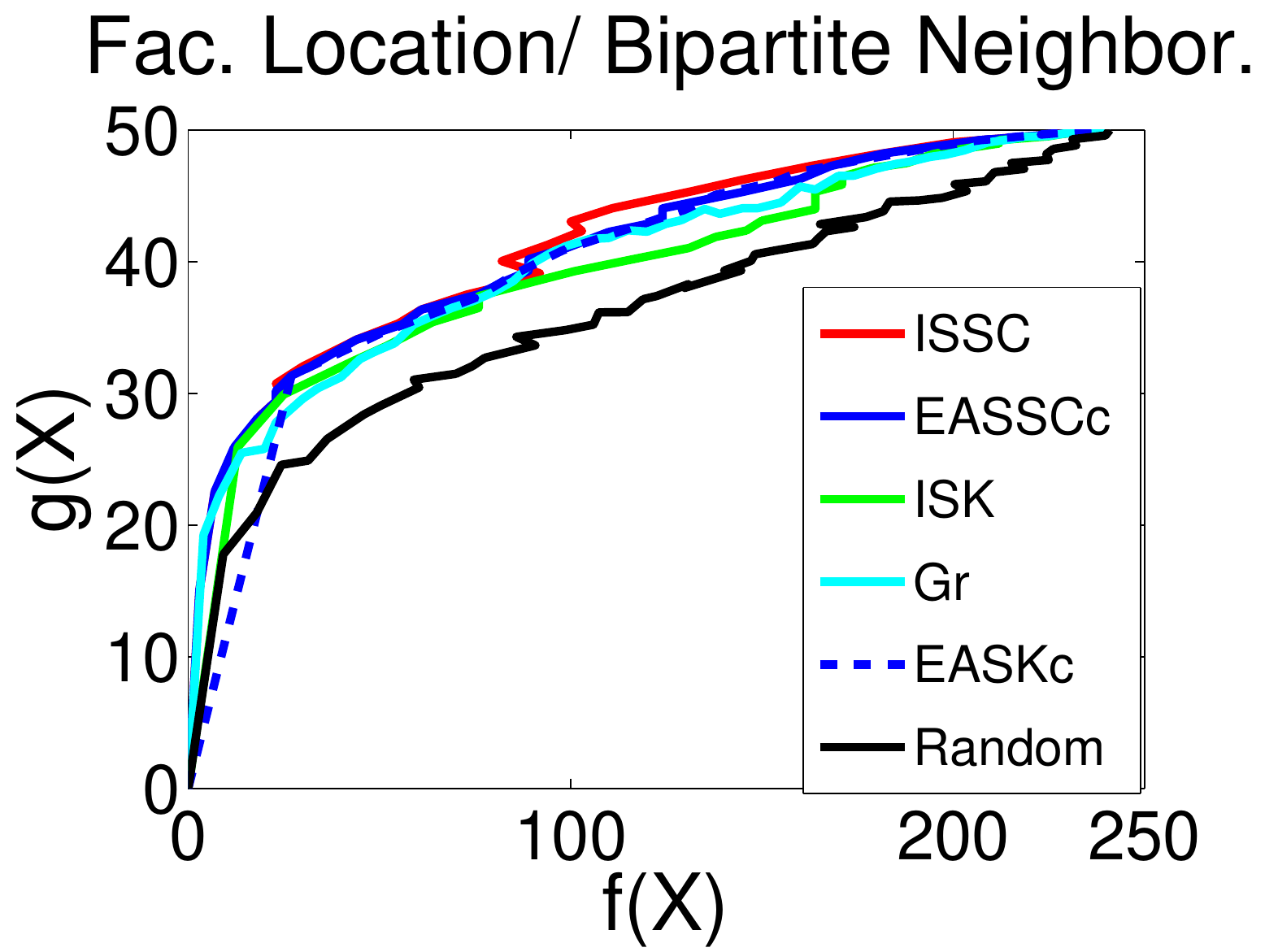}\hspace{-10pt} 
  ~ 
\includegraphics[width=0.45\textwidth]{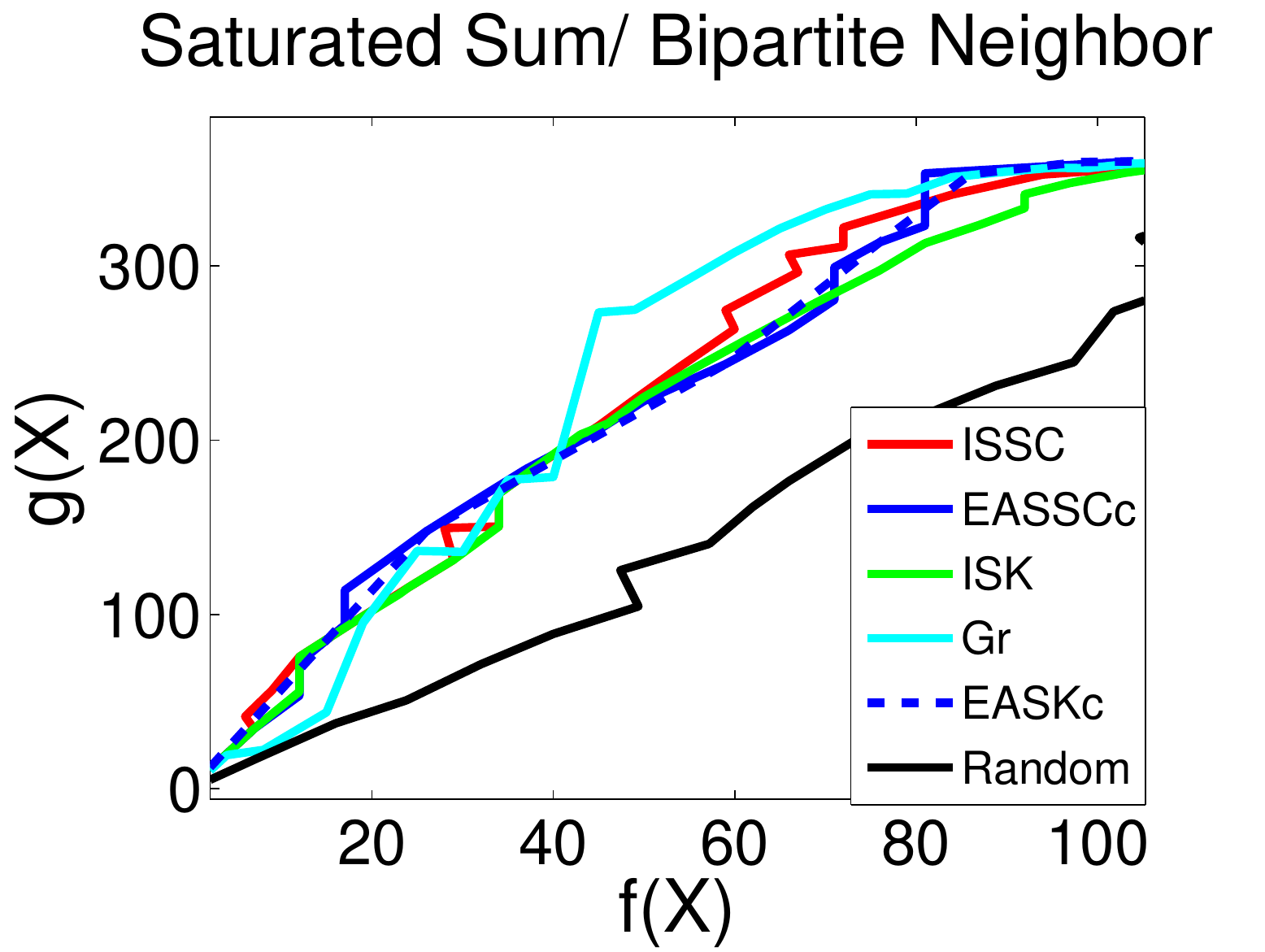} \hspace{-10pt} 
\end{minipage}
\caption{\small{Comparison of the algorithms in the text.}}
\label{fig:sim}
\end{wrapfigure}}

We compare our different algorithms on Problems 1 and 2 with $f$
being the bipartite neighborhood and $g$ being the facility location
and saturated sum respectively. \arxiv{Since the primal and dual variants of the submodular set cover problem are similar, we just
use the primal variants of ISSC and
EASSC.} Furthermore, in our
experiments, we observe that the neighborhood function $f$ has a
curvature $\curvf{f} = 1$.  
%\JTR{Above, the simpler expression for EA
%  comes about when $\curvf{f} = 0$ but something must be wrong
%  here. Double check this. Also see the comment tagged 5XXXXX above.}\RTJ{This issue is solved.}
Thus, it suffices to use the simpler versions of algorithm EA (i.e.,
algorithm EASSCc and EASKc). The results are shown in
Figure~\ref{fig:sim}. We observe that on the real-world instances, all
our algorithms perform almost comparably. This implies, moreover, that
the iterative variants, viz. Gr, ISSC and ISK, perform comparably to
the more complicated EA-based ones, although EASSC and EASK have
better theoretical guarantees. We also compare against a baseline of
selecting random sets (of varying cardinality), and we see that our
algorithms all perform much better.  In terms of the running time,
computing the Ellipsoidal Approximation for $|\Gamma(X)|$ with $|V| =
50$ takes about $5$ hours while all the iterative variants (i.e., Gr,
ISSC and ISK) take less than a second. This difference is much more
prominent on larger instances (for example $|V| = 500$). \looseness-1

\JTR{There needs to be an additional very simple discussion here on
  one-stop-shopping: i.e., you have a problem like Problem 1 or
  Problem 2. Based on the theoretical and empirical results we've got,
  state what we recommend for each of Problems 1 and 2.
  Looks like Gprimal for problem 1 and ISSC for problem 2, but maybe
a bit more empirical results could strengthen this recommendation.}\RTJ{Yes, I will be rerunning the experiments and will modify this then.}

\section{Discussions\arxiv{ and related work}}
In this paper, we propose a unifying framework for problems 1 and 2 based on suitable surrogate functions. We provide a number of iterative algorithms which are very practical and scalable (like Gr, ISK and ISSC), and also algorithms like EASSC and EASK, which though more intensive, obtain tight approximation bounds. Finally, we empirically compare our algorithms, and show that the iterative algorithms compete empirically with the more complicated and theoretically better approximation algorithms.\arxiv{

To our
knowledge, this paper provides the first general framework of approximation
algorithms for Problems 1 and 2 for monotone submodular functions $f$ and $g$. A number of papers, however, investigate related problems and approaches. For example,~\cite{fujishige2005submodular} investigates an exact algorithm for solving problem 1, with equality instead of inequality. However, since problem 1 subsumes the problem of minimizing a monotone submodular function subject to a cardinality equality constraint, and is hence NP hard~\cite{nagano2011}. Hence this algorithm in the worst case, must be exponential. Furthermore, a similar problem was considered in~\cite{krause06near} with one specific instance of a function $f$, which is not submodular. They also use, a considerably different algorithm. Also, an algorithm equivalent to the first iteration of ISSC was proposed in~\cite{wan2010greedy, du2011minimum} and ISSC not only generalizes this, but we also provide a more explicit approximation guarantee (we provide an elaborate discussion on this in the section describing ISSC). We also point out that, a special case of SCSK was considered in~\cite{lin2012submodularity}, with $f$ being submodular, and $g$ modular (we called this the \emph{submodular span problem}). The authors there use an algorithm very similar to Algorithm~\ref{alg:alg1}, to convert this problem into an instance of minimizing a submodular function subject to a knapsack constraint, for which they use the algorithm of~\cite{svitkina2008submodular}. Unfortunately, the algorithm of~\cite{svitkina2008submodular} does not scale very well. Our algorithms for this problem, on the other hand, would continue to scale very well in practice.

Similarly, a number of approximation algorithms
have been shown for Problem 0~\cite{rkiyeruai2012, narasimhanbilmes,
  kawahara2011prismatic}. The algorithms in~\cite{rkiyeruai2012,
  narasimhanbilmes} are scalable and practical, but lack theoretical
guarantees. The algorithm of~\cite{kawahara2011prismatic} though exact, employs a branch and bound technique which is often inefficient in
practice (the timing analysis from~\cite{kawahara2011prismatic} also
depicts that). These facts are not surprising, since problem 0 is not only NP hard but also inapproximable. Moreover, these algorithms are not comparable to ours,
since we directly model the hard constraints and our bicriteria
results give worst-case bounds on the deviation from the constraints
and the optimal solution. This is often important, since there are hard constraints in many practical applications (in the form of power constraints, or budget constraints). Casting it as Problem 0, however, no longer
has guarantees on the deviation from the constraints.

} For future work, we would like to empirically evaluate our algorithms on many of the real world problems described above, particularly the limited vocabulary data subset selection application for speech corpora, and the machine translation application.

{\bf Acknowledgments:} Special thanks to Kai Wei and Stefanie Jegelka for discussions, to Bethany Herwaldt for going through an early draft of this manuscript and to the anonymous reviewers for useful reviews. This material is based upon work supported by the National Science
Foundation under Grant No. (IIS-1162606), a
Google and a Microsoft award, and
by the Intel Science and
Technology Center for Pervasive Computing.\looseness-1
\notarxiv{\small}
\bibliographystyle{abbrv}
\bibliography{../Combined_Bib/submod}

\begin{thebibliography}{10}

\bibitem{atamturk2009submodular}
A.~Atamt{\"u}rk and V.~Narayanan.
\newblock The submodular knapsack polytope.
\newblock {\em Discrete Optimization}, 2009.

\bibitem{boykovJolly01}
Y.~Boykov and M.~Jolly.
\newblock Interactive graph cuts for optimal boundary and region segmentation
  of objects in n-d images.
\newblock In {\em ICCV}, 2001.

\bibitem{conforti1984submodular}
M.~Conforti and G.~Cornuejols.
\newblock Submodular set functions, matroids and the greedy algorithm: tight
  worst-case bounds and some generalizations of the {R}ado-{E}dmonds theorem.
\newblock {\em Discrete Applied Mathematics}, 7(3):251--274, 1984.

\bibitem{cun82}
W.~H. Cunningham.
\newblock Decomposition of submodular functions.
\newblock {\em Combinatorica}, 3(1):53--68, 1983.

\bibitem{du2011minimum}
H.~Du, W.~Wu, W.~Lee, Q.~Liu, Z.~Zhang, and D.-Z. Du.
\newblock On minimum submodular cover with submodular cost.
\newblock {\em Journal of Global Optimization}, 50(2):229--234, 2011.

\bibitem{feige1998threshold}
U.~Feige.
\newblock A threshold of ln n for approximating set cover.
\newblock {\em Journal of the ACM (JACM)}, 1998.

\bibitem{fiege2011submodmax}
U.~Feige, V.~Mirrokni, and J.~Vondr{\'a}k.
\newblock Maximizing non-monotone submodular functions.
\newblock {\em SIAM J. COMPUT.}, 40(4):1133--1155, 2011.

\bibitem{fujishige2005submodular}
S.~Fujishige.
\newblock {\em Submodular functions and optimization}, volume~58.
\newblock Elsevier Science, 2005.

\bibitem{timit}
J.~Garofolo, F.~Lamel, L., J.~W., Fiscus, D.~Pallet, and N.~Dahlgren.
\newblock Timit, acoustic-phonetic continuous speech corpus.
\newblock In {\em DARPA}, 1993.

\bibitem{goel2009approximability}
G.~Goel, C.~Karande, P.~Tripathi, and L.~Wang.
\newblock Approximability of combinatorial problems with multi-agent submodular
  cost functions.
\newblock In {\em FOCS}, 2009.

\bibitem{goemans2009approximating}
M.~Goemans, N.~Harvey, S.~Iwata, and V.~Mirrokni.
\newblock Approximating submodular functions everywhere.
\newblock In {\em SODA}, pages 535--544, 2009.

\bibitem{guillory2010interactive}
A.~Guillory and J.~Bilmes.
\newblock Interactive submodular set cover.
\newblock {\em In ICML}, 2010.

\bibitem{guillory2011simultaneous}
A.~Guillory and J.~Bilmes.
\newblock Simultaneous learning and covering with adversarial noise.
\newblock In {\em ICML}, 2011.

\bibitem{rkiyeruai2012}
R.~Iyer and J.~Bilmes.
\newblock Algorithms for approximate minimization of the difference between
  submodular functions, with applications.
\newblock {\em In UAI}, 2012.

\bibitem{rkiyersubmodBregman2012}
R.~Iyer and J.~Bilmes.
\newblock The submodular {B}regman and {L}ov\'asz-{B}regman divergences with
  applications.
\newblock In {\em NIPS}, 2012.

\bibitem{iyermirrordescent}
R.~Iyer, S.~Jegelka, and J.~Bilmes.
\newblock {Mirror descent like algorithms for submodular optimization}.
\newblock {\em NIPS Workshop on Discrete Optimization in Machine Learning
  (DISCML)}, 2012.

\bibitem{curvaturemin}
R.~Iyer, S.~Jegelka, and J.~Bilmes.
\newblock {Curvature and Optimal Algorithms for Learning and Minimizing
  Submodular Functions }.
\newblock In {\em NIPS}, 2013.

\bibitem{rkiyersemiframework2013}
R.~Iyer, S.~Jegelka, and J.~Bilmes.
\newblock Fast semidifferential based submodular function optimization.
\newblock In {\em ICML}, 2013.

\bibitem{jegelkacvpr}
S.~Jegelka and J.~Bilmes.
\newblock Submodularity beyond submodular energies: coupling edges in graph
  cuts.
\newblock In {\em Computer Vision and Pattern Recognition (CVPR)}, 2011.

\bibitem{jegelka2011-nonsubmod-vision}
S.~Jegelka and J.~A. Bilmes.
\newblock Submodularity beyond submodular energies: coupling edges in graph
  cuts.
\newblock In {\em CVPR}, 2011.

\bibitem{jegelkanips}
S.~Jegelka, H.~Lin, and J.~Bilmes.
\newblock On fast approximate submodular minimization.
\newblock {\em In NIPS}, 2011.

\bibitem{kawahara2011prismatic}
Y.~Kawahara and T.~Washio.
\newblock Prismatic algorithm for discrete dc programming problems.
\newblock In {\em NIPS}, 2011.

\bibitem{kellerer2004knapsack}
H.~Kellerer, U.~Pferschy, and D.~Pisinger.
\newblock {\em Knapsack problems}.
\newblock Springer Verlag, 2004.

\bibitem{krause05note}
A.~Krause and C.~Guestrin.
\newblock A note on the budgeted maximization on submodular functions.
\newblock Technical Report CMU-CALD-05-103, Carnegie Mellon University, 2005.

\bibitem{krause06near}
A.~Krause, C.~Guestrin, A.~Gupta, and J.~Kleinberg.
\newblock Near-optimal sensor placements: Maximizing information while
  minimizing communication cost.
\newblock In {\em IPSN}, 2006.

\bibitem{krause08robust}
A.~Krause, B.~McMahan, C.~Guestrin, and A.~Gupta.
\newblock Robust submodular observation selection.
\newblock {\em Journal of Machine Learning Research (JMLR)}, 9:2761--2801,
  2008.

\bibitem{krause2008near}
A.~Krause, A.~Singh, and C.~Guestrin.
\newblock Near-optimal sensor placements in {G}aussian processes: Theory,
  efficient algorithms and empirical studies.
\newblock {\em JMLR}, 9:235--284, 2008.

\bibitem{kulik2009maximizing}
A.~Kulik, H.~Shachnai, and T.~Tamir.
\newblock Maximizing submodular set functions subject to multiple linear
  constraints.
\newblock In {\em SODA}, 2009.

\bibitem{lin2012submodularity}
H.~Lin.
\newblock {\em Submodularity in Natural Language Processing: Algorithms and
  Applications}.
\newblock PhD thesis, University of Washington, Dept.\ of EE, 2012.

\bibitem{lin2009select}
H.~Lin and J.~Bilmes.
\newblock How to select a good training-data subset for transcription:
  Submodular active selection for sequences.
\newblock In {\em Interspeech}, 2009.

\bibitem{linbudget}
H.~Lin and J.~Bilmes.
\newblock Multi-document summarization via budgeted maximization of submodular
  functions.
\newblock {\em In NAACL}, 2010.

\bibitem{lin2011-class-submod-sum}
H.~Lin and J.~Bilmes.
\newblock A class of submodular functions for document summarization.
\newblock In {\em The 49th Meeting of the Assoc.\ for Comp.\ Ling.\ Human
  Lang.\ Technologies (ACL/HLT-2011)}, Portland, OR, June 2011.

\bibitem{lin11}
H.~Lin and J.~Bilmes.
\newblock Optimal selection of limited vocabulary speech corpora.
\newblock In {\em Interspeech}, 2011.

\bibitem{lovasz1983}
L.~Lov\'asz.
\newblock Submodular functions and convexity.
\newblock {\em Mathematical Programming}, 1983.

\bibitem{moore2010intelligent}
R.~C. Moore and W.~Lewis.
\newblock Intelligent selection of language model training data.
\newblock In {\em Proceedings of the ACL 2010 Conference Short Papers}, pages
  220--224. Association for Computational Linguistics, 2010.

\bibitem{nagano2011}
K.~Nagano, Y.~Kawahara, and K.~Aihara.
\newblock Size-constrained submodular minimization through minimum norm base.
\newblock In {\em ICML}, 2011.

\bibitem{narasimhanbilmes}
M.~Narasimhan and J.~Bilmes.
\newblock A submodular-supermodular procedure with applications to
  discriminative structure learning.
\newblock In {\em UAI}, 2005.

\bibitem{nemhauser78}
G.~Nemhauser and L.~Wolsey.
\newblock Best algorithms for approximating the maximum of a submodular set
  function.
\newblock {\em Mathematics of Operations Research}, 3(3):177--188, 1978.

\bibitem{nikolova2010approximation}
E.~Nikolova.
\newblock Approximation algorithms for offline risk-averse combinatorial
  optimization, 2010.

\bibitem{rother2006cosegmentation}
C.~Rother, T.~Minka, A.~Blake, and V.~Kolmogorov.
\newblock Cosegmentation of image pairs by histogram matching-incorporating a
  global constraint into {MRF}s.
\newblock In {\em CVPR}, volume~1, pages 993--1000. IEEE, 2006.

\bibitem{rousu2006efficient}
J.~Rousu and J.~Shawe-Taylor.
\newblock Efficient computation of gapped substring kernels on large alphabets.
\newblock {\em Journal of Machine Learning Research}, 6(2):1323, 2006.

\bibitem{sviridenko2004note}
M.~Sviridenko.
\newblock A note on maximizing a submodular set function subject to a knapsack
  constraint.
\newblock {\em Operations Research Letters}, 32(1):41--43, 2004.

\bibitem{svitkina2008submodular}
Z.~Svitkina and L.~Fleischer.
\newblock Submodular approximation: Sampling-based algorithms and lower bounds.
\newblock In {\em FOCS}, pages 697--706, 2008.

\bibitem{vazirani2004approximation}
V.~V. Vazirani.
\newblock {\em Approximation algorithms}.
\newblock springer, 2004.

\bibitem{vondrak2010submodularity}
J.~Vondr{\'a}k.
\newblock Submodularity and curvature: the optimal algorithm.
\newblock {\em RIMS Kokyuroku Bessatsu}, 23, 2010.

\bibitem{wan2010greedy}
P.-J. Wan, D.-Z. Du, P.~Pardalos, and W.~Wu.
\newblock Greedy approximations for minimum submodular cover with submodular
  cost.
\newblock {\em Computational Optimization and Applications}, 45(2):463--474,
  2010.

\bibitem{wolsey1982analysis}
L.~A. Wolsey.
\newblock An analysis of the greedy algorithm for the submodular set covering
  problem.
\newblock {\em Combinatorica}, 2(4):385--393, 1982.

\end{thebibliography}
\end{document}